\documentclass[12pt,draftclsnofoot,onecolumn]{IEEEtran}

\usepackage{indentfirst,flushend}
\usepackage[justification=centering]{caption}
\usepackage{amssymb,amscd,amsmath,amsthm,amsfonts,mathrsfs,times,bm,bbm}
\usepackage[colorlinks,linkcolor=black,anchorcolor=black,citecolor=black]{hyperref}
\usepackage{latexsym}
\usepackage{cite}
\usepackage{graphicx}
\usepackage{subfigure}
\usepackage{color}
\usepackage{cases}
\usepackage{algorithm}
\usepackage{algorithmic}
\usepackage{clrscode}
\usepackage{psfrag,stfloats,nicefrac}
\usepackage{supertabular}
\usepackage{makecell}
\usepackage{booktabs}
\usepackage{array}
\usepackage{multirow}
\usepackage{threeparttable}
\allowdisplaybreaks[4]

\newtheorem{remark}{Remark}
\newtheorem{theorem}{Theorem}

\newtheorem{proposition}{Proposition}

\newcommand{\rmnum}[1]{\uppercase\expandafter{\romannumeral #1\relax}}

\begin{document}
\title{Cross-Layer Optimization: Joint User Scheduling and Beamforming Design With QoS Support in Joint Transmission Networks}
\author{Shiwen~He,~\IEEEmembership{Member,~IEEE},~Zhenyu An,~\IEEEmembership{Student Member,~IEEE},\\~Jianyue Zhu,~\IEEEmembership{Member,~IEEE},~Min Zhang,~Yongming Huang,~\IEEEmembership{Senior Member,~IEEE},~and Yaoxue~Zhang,~\IEEEmembership{Senior Member,~IEEE}
\thanks{S. He is with the School of Computer Science and Engineering, Central South University, Changsha 410083, China. S. He is also with the National Mobile Communications Research Laboratory, Southeast University, and the Purple Mountain Laboratories, Nanjing 210096, China. (email: shiwen.he.hn@csu.edu.cn). }
\thanks{Z. An is with the Purple Mountain Laboratories, Nanjing 210096, China. (email: anzhenyu@pmlabs.com.cn). }
\thanks{J. Zhu is with  the College of Electronic and Information Engineering, Nanjing University of Information Science and Technology, Nanjing, China,  She is also with the School of Information Science and Engineering, Southeast University, Nanjing 210096, China.(email: zhujy@seu.edu.cn).}
\thanks{M. Zhang is with the School of information and communication, Hunan Post and Telecommunication College, Hunan, Changsha, 410015, China. (email: 380072457@qq.com)}
\thanks{Y. Huang is with the National Mobile Communications Research Laboratory, School of Information Science and Engineering, Southeast University, Nanjing 210096, China. He is also with the Purple Mountain Laboratories, Nanjing 210096, China. (email: huangym@seu.edu.cn). }
\thanks{Y. Zhang is with the Department of Computer Science and Technology, Tsinghua University, Beijing 100084, China. (email: zhangyx@tsinghua.edu.cn)}
}

\maketitle
\vspace{-.6 in}

\begin{abstract}
User scheduling and beamforming design are two crucial yet coupled topics for wireless communication systems. They are usually optimized separately with conventional optimization methods. In this paper, a novel cross-layer optimization problem is considered, namely, the user scheduling and beamforming are jointly discussed subjecting to the requirement of per-user quality of service (QoS) and the maximum allowable transmit power for multicell multiuser joint transmission networks. To achieve the goal, a mixed discrete-continue variables combinational optimization problem is investigated with aiming at maximizing the sum rate of the communication system. To circumvent the original non-convex problem with dynamic solution space, we first transform it into a 0-1 integer and continue variables optimization problem, and then obtain a tractable form with continuous variables by exploiting the characteristics of 0-1 constraint. Finally, the scheduled users and the optimized beamforming vectors are simultaneously calculated by an alternating optimization algorithm. We also theoretically prove that the base stations allocate zero power to the unscheduled users. Furthermore,  two heuristic optimization algorithms are proposed respectively based on brute-force search and greedy search. Numerical results validate the effectiveness of our proposed methods, and the optimization approach gets relatively balanced results compared with the other two approaches.
\end{abstract}

\section{\label{Introduction}Introduction}
At present, the application of the fifth generation (5G) communication technologies is gradually infiltrated into people's daily life. However, to satisfy the future demands for information and communications technology (ICT) in 2030~\cite{OJWCJiang2021}, the researchers in both academia and industry have turned their attention to the research of 6G communication technologies. Numerous potential communication technologies, e.g., massive multiple-input multiple-output (MIMO) and millimeter-wave communication, may be adopted to satisfy the extreme demands of future wireless traffic, such as in the ultra-reliable low latency communication and massive machine type communication scenarios~\cite{ProHe2021}. In particular, among the future communication technologies, coordinated  multi-point joint transmission~\cite{MagIrmer2011} and multi-connected technologies~\cite{ProTataria2021} are two enablers to address the massive user connection problem, which have attracted extensive attentions in both academia and industry, to address the massive user connection problem for ultra-dense networks.

For coordinated multi-point joint transmission, the fundamental issues are user scheduling and beamforming design implemented at the media access control layer~\cite{TSPDimic2005} and the physical layer~\cite{TWCZhang2009}, respectively. Unfortunately, for the design of transmission scheme in wireless communication systems, these two issues are always coupled, which is difficult to be solved. To the best of our knowledge, in the literature, one usually solved a serious of subproblems to release the coupled relations. For example, the authors of~\cite{JSACYoo2006} proposed a semiorthogonal user selection (SUS) algorithm cooperating with the zero-forcing beamforming (ZFBF) for the downlink multiuser communication systems. Numerical results show that ZFBF-SUS transmission performs reasonably well under practical value of users' number. However, the authors of~\cite{JSACYoo2006} did not consider the quality of service (QoS) in the ZFBF-SUS algorithm. On the other hand, for a fixed scheduled user set, the optimization of the transceivers is also a research hotspot for wireless communication systems~\cite{TSPYu2007}. In general, for multicell multiuser MIMO communication networks, even fixing the scheduled user set, the problem of beamforming design is non-convex and is hard to be solved~\cite{TITHuh2012}. Usually, the uplink-downlink duality theory insert a space is used as a powerful tool for tackling the non-convex transceivers design~\cite{TWCDah2010,TITZhang2012,TCOMHe2015}. Note that in these mentioned references, the user scheduling and beamforming design are separately considered. To further improve the performance of communication systems, recently, cross-layer design is increasingly becoming popular~\cite{CSTFu2014}. For example, the authors of~\cite{CLSun2021} investigated various resource allocation and user scheduling methods under the constraint of delivery latency for the downlink multiuser communication systems. The authors of~\cite{TWCNasir2021} studied the resource allocation problem for ultra-reliable low latency communication systems. However, the authors of~\cite{CLSun2021} and~\cite{TWCNasir2021} only studied the resource allocation problem for a single antenna multiuser communication system.

In the last few years, joint user scheduling and beamforming design has became another hotspot in both academia and industry. The authors of~\cite{JSACHong2013} studied the joint base station (BS) clustering and beamforming design for partial coordinated transmission in heterogeneous networks. The authors of~\cite{TWCZhang2017} investigated the joint user scheduling and beamforming design for large-scale multiple-input single-output (MISO) systems. However, the beamforming design and user scheduling are separately optimized, which cannot guarantee the optimal performance of wireless communication systems. The authors of~\cite{TWCJiang2018} investigated the joint user scheduling and analog beam selection problem for codebook-based massive MIMO downlink systems with hybrid antenna architecture and a diagonal baseband precoding matrix. The authors of~\cite{TCOMAntoniolo2020} investigated the problem of joint user scheduling and beamforming design via adopting rate-relaxation variables for the downlink of coordinated multicell multiuser communication systems. But, the provided method allocates a zero power to the deactivated users, which was not proved theoretically. Furthermore, as shown in Fig.~4, the BS may allocate transmitting power to a deactivated user, i.e., whose rate requirement has not been met. The authors of~\cite{TWCKhan2020} investigated the joint user scheduling and beamforming design for coordinated beamforming communication system without the requirements of quality-of-service (QoS) per-user. The authors of~\cite{SysAkhtar2021} studied the joint user scheduling and antenna selection for the downlink multiuser communication system with zero-forcing beamforming and the requirement of QoS. To address the problem of joint user scheduling and beamforming design in ultra-dense communication networks, the authors of~\cite{CLHe2021} aimed to maximizing the set cardinality of scheduled users subjecting to the requirement of QoS and the maximum allowable transmit power. 

In this paper, as a new work for the downlink of multicell multiuser joint transmission networks, we focus on investigating the problem of joint user scheduling and beamforming design subjecting to the requirement of  per-user QoS and the maximum allowable transmit power. The problem of joint user scheduling and beamforming design is formulated as a mixed discrete-continue variables combinational optimization problem. To overcome the difficulties encountered in solving the problem, some basic problem transformation methods are derived and then an effective and efficient optimization method is developed. The main contributions are listed as follows:
\begin{itemize}
\item Firstly, to solve the mixed integer programming problem, we transform this problem into a 0-1 integer and continue variables optimization problem. Particularly, we theoretically prove that the BSs allocate zero power to the unscheduled users.
\item Secondly, by exploiting the characteristics of 0-1 constraints,  a tractable form is obtained with  continuous variables.
\item Thirdly, an effective and efficient optimization algorithm is developed. The convergence of the proposed algorithm is also guaranteed with monotonic boundary theory and sub-gradient theory.
\item Fourthly, to compare the performance, two heuristic optimization algorithms are also developed with brute-force search and greet search, respectively.
\item Finally, a large number of experimental results are provided to validate the effectiveness of the developed algorithm.
\end{itemize}

The rest of this paper is organized as follows. Section~\rmnum{2} describes the system model and raises the joint user scheduling and beamforming problem. Section~\rmnum{3} formulates the problem transformation to obtain a tractable form. Section~\rmnum{4} proposes an alternating optimization algorithm and two heuristic search algorithms. Numerical results are presented in Section~\rmnum{5}. Finally, conclusions are drawn in Section~\rmnum{6}.

\textbf{\textcolor{black}{$\mathbf{\mathit{Notations}}$}}: We indicate matrices by bold uppercase letters and column vectors by bold lowercase letters. $\mathcal{A}$ denotes a set and $\mathbf{A}[m,n]$ denotes the element of the $m$-th row and the $n$-th column in matrix $\mathbf{A}$. $\left|\cdot\right|$ denotes the absolute value of a complex scalar or the cardinality of a set, and $\left\|\cdot\right\|$ denotes the Euclidean vector norm. $(\cdot)^T$, $(\cdot)^H$ and $(\cdot)^{-1}$ indicate the transpose, the Hermitian transpose, and the inverse of a matrix, respectively. $\mathbb{R}^M$ indicates $M$-dimensional real space and $\mathbb{C}^M$ indicates $M$-dimensional complex space.

\section{\label{SystemModelAndProblem}System Model and Problem Formulation}
In this work, we investigate the coordinated user scheduling and beamforming design for the downlink of multicell multiuser joint transmission networks, where a $B$ BSs cooperative cluster serving $S$ single antenna users. Each BS is equipped with $N_{\mathrm{t}}$ transmitting antennas. The $S$ users are scheduled from $K$ users waiting to be served. Let $\mathcal{B}=\{1,2,\cdots,B\}$, $\mathcal{K}=\{1,2,\cdots,K\}$, and $\mathcal{S}\subseteq\mathcal{K}$ be the set of all BSs, candidate users, and served users, respectively. The channel coefficient between the $b$-th BS and the $k$-th user is denoted as $\mathbf{h}_{k,b}=\sqrt{\varrho_{k,b}}\widetilde{\mathbf{h}}_{k,b}\in\mathbb{C}^{N_{\mathrm{t}}\times 1}$, where the channel power $\varrho_{k,b}$ follows the large scale fading characteristic and it is given ${\varrho _{k,b}} = 10\hat{} ~(( - 38{\log _{10}}({d_{k,b}}) - 34.5 + {\varpi _{k,b}})/10)$ with $d_{k,b}$ being the distance between the $b$-th BS and the $k$-th user, ${\varpi _{k,b}}$ representing the log-normal shadow fading with zero mean and standard deviation 8 dB \cite{3GPP}. The elements of $\widetilde{\mathbf{h}}_{k,b}$ are independent and identically distributed (i.i.d.) with $\mathcal{CN}\left(0,1\right)$. Let $\mathbf{h}_{k}\in\mathbb{C}^{BN_{\mathrm{t}}\times 1}$ be the cascaded channel coefficient, i.e., $\mathbf{h}_{k}=\left[\mathbf{h}_{k,1}^{\mathrm{H}}, \mathbf{h}_{k,2}^{\mathrm{H}},\cdots,\mathbf{h}_{k,B}^{\mathrm{H}}\right]^{\mathrm{H}}$.  Let $p_{k}\in\mathbb{R}_{+}$ and $\mathbf{w}_{k}=\left[\mathbf{w}_{k,1}^{\mathrm{H}}, \mathbf{w}_{k,2}^{\mathrm{H}},\cdots,\mathbf{w}_{k,B}^{\mathrm{H}}\right]^{\mathrm{H}}\in\mathbb{C}^{BN_{\mathrm{t}}\times 1}$ respectively represent the coefficient of transmitting power and the coordinated unit-norm beamforming vector used by the BSs, where $\mathbf{w}_{k,b}\in\mathbb{C}^{N_{\mathrm{t}}\times 1}$ is the beamforming vectors used by the $b$-th BS for the $k$-th user. Thus, the received baseband signal at the $k$-th user is given by
\begin{equation}\label{CUSBD01}
y_{k}=\sum\limits_{l\in\mathcal{S}}\sqrt{p_{l}}\mathbf{h}_{k}^{H}\mathbf{w}_{l}s_{l}+n_{k}, \forall k\in\mathcal{S},
\end{equation}
where $s_{k}$ and $n_{k}$ denote the baseband signal and the additive white Gaussian noise with $\mathcal{CN}\left(0,\sigma_{k}^{2}\right)$ at the $k$-th user, respectively. Let $\overline{\mathbf{h}}_{k}=\frac{\mathbf{h}_{k}}{\sigma_{k}}$, the signal-to-interference-plus-noise ratio (SINR) of the $k$-th user is formulated as follows
\begin{equation}\label{CUSBD02}
\gamma_{k}=\frac{p_{k}\left|\overline{\mathbf{h}}_{k}^{H}\mathbf{w}_{k}\right|^{2}}
{\sum\limits_{l\neq k,l\in\mathcal{S}}p_{l}\left|\overline{\mathbf{h}}_{k}^{H}\mathbf{w}_{l}\right|^{2}+1}.
\end{equation}
Thus, the achievable rate of the $k$-th user is expressed as $R_{k}=\log_{2}\left(1+\gamma_{k}\right)$.

In this work, the objective is to maximize the system sum rate via coordinated user scheduling and beamforming design. Accordingly, the optimization problem is formulated as
\begin{subequations}\label{CUSBD03}
\begin{align}
&\max_{\mathcal{S}\subseteq\mathcal{K}, \left\{p_{k},\mathbf{w}_{k}\right\}} \sum\limits_{k\in\mathcal{S}}R_{k}, \label{CUSBD03a}\\
\mathrm{s.t.}~&p_{k}>0, \left\|\mathbf{w}_{k}\right\|_{2}=1, \forall k\in\mathcal{S},\label{CUSBD03b}\\
&r_{k}\leq R_{k}, \forall k\in\mathcal{S}, \left|\mathcal{S}\right|\leq BN_{\mathrm{t}},\label{CUSBD03c}\\
&\sum\limits_{k\in\mathcal{S}}p_{k}\left\|\mathbf{Q}_{b}\mathbf{w}_{k}\right\|_{2}^{2}\leq P_{b}, \forall b\in\mathcal{B},\label{CUSBD03d}
\end{align}
\end{subequations}
where $P_{b}$ is the maximum allowable transmit power of the $b$-th BS, $r_{k}>0$ is the required minimum user rate of the $k$-th user, $\forall k\in\mathcal{S}$, and permutation matrix $\mathbf{Q}_{b}$ is defined as
\begin{equation}\label{CUSBD04}
\mathbf{Q}_{b}\left[i,j\right]=
\begin{cases}
1,~~i=j=\left(b-1\right)N_{\mathrm{t}}+1,\cdots,bN_{\mathrm{t}}\\
0,~~\text{otherwise} .
\end{cases}
\end{equation}
In problem~\eqref{CUSBD03}, constraint~\eqref{CUSBD03c} assures that the minimum user rate demand of each user is to be satisfied and the total number of served users is not lager than the total number of transmitting antennas. Constraint~\eqref{CUSBD03d} ensures that the transmitting power of per-BS is not overpass the maximum allowable transmitting power.

It is not difficulty to see that problem~\eqref{CUSBD03} is a mixed integer programming problem. Solving problem~\eqref{CUSBD03} needs to address the joint optimization of scheduled user set $\mathcal{S}$ selection, the beamforming design, and the power allocation, which is NP-hard. Accordingly, for problem~\eqref{CUSBD03}, the optimal solution, even the local optimal solution, is difficult to obtain. Specifically, the difficulty of problem~\eqref{CUSBD03} lies in the following obstacles. Firstly, the achievable user rate $R_{k}$ is a non-convex function with respect to $\left\{p_{k},\mathbf{w}_{k}\right\}$ due to the interference channels in joint transmission systems. Secondly, constraint~\eqref{CUSBD03c} is non-convex due to the uncertain set $\mathcal{S}$, i.e., the set of users scheduled. Thirdly, constraint~\eqref{CUSBD03d} is non-convex with respect to $\left\{p_{k},\mathbf{w}_{k}\right\}$. Furthermore, another problem about feasibility is necessary to be discussed, i.e., whether each user in the set $\mathcal{K}$ satisfies the minimum user rate requirement $r_{k}\leq R_{k}, \forall k\in\mathcal{K}$ or not. For the joint transmission networks, we can adopt the single-user communication with full power maximum ratio transmission to simply check the feasiblity of problem~\eqref{CUSBD03}. 
\begin{figure}[t]
\centering
\includegraphics[width=1\columnwidth,keepaspectratio]{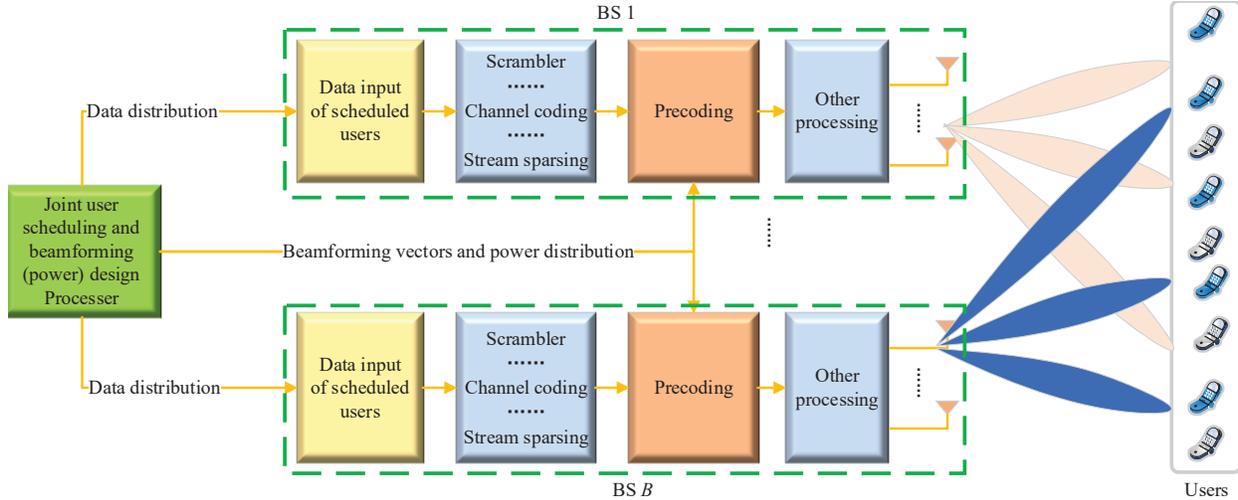}
\caption{System model of the downlink joint transmission network.}
\label{SystemModel}
\end{figure}

In general, the transmission scheme is designed for multiuser communication systems with a fixed set of users scheduled. Differently, in this work, the set of users scheduled is also needed to be optimized. This implies that the user scheduling and the design of transmission scheme are jointly optimized, as illustrated in Fig.~\ref{SystemModel}.  The transmission scheme are actually  jointly optimized cross the physical layer and media access control layer for wireless communication systems. Assume that the minimum rate requirement for each user $r_{k}\leq R_{k}$ can be met with the maximum ratio transmission and full power transmission, i.e., single user communication without considering any interference except for the additive white Gaussian noise, $\forall k\in\mathcal{K}$. The number of possible set $\mathcal{S}$ is $\sum\limits_{\left| {\cal S} \right| = 1}^{\min (B{N_{\rm{t}}},K)} {\frac{{K!}}{{\left| {\cal S} \right|!\left( {K - \left| {\cal S} \right|} \right)!}}} $. Accordingly, if there exists an optimization algorithm for solving problem~\eqref{CUSBD03} with a given set $\mathcal{S}$ of users scheduled, thus, this algorithm can be used to solve $\sum\limits_{\left| {\cal S} \right| = 1}^{\min (B{N_{\rm{t}}},K)} {\frac{{K!}}{{\left| {\cal S} \right|!\left( {K - \left| {\cal S} \right|} \right)!}}} $ combinational problems and the optimum combination can be selected out to achieve the maximum sum rate. However, in the sequel, we do not assume that $\forall k\in\mathcal{K}$, $r_{k}\leq R_{k}$ holds, implying the difficulty of the proposed problem.

\section{\label{ProblemTransformation} Problem Transformations}
In this section, we focus on releasing the uncertain set $\mathcal{S}$ of users scheduled and looking for a dual description of problem~\eqref{CUSBD03}. Consequently, a slightly tractable form of problem~\eqref{CUSBD03} is obtained and induces the design of optimization algorithms.

\subsection{Release of set $\mathcal{S}$}
For simplicity, let $\bm{\mu}=\left\{\mu_{1},\mu_{2},\cdots,\mu_{k-1},\mu_{k},\mu_{k+1},\cdots,\mu_{K}\right\}$ and $\mu_{k}$ is a binary variable, i.e., $\mu_{k}\in\left\{0,1\right\}$. If the $k$-th user is scheduled to be served for the joint transmission, $\mu_{k}=1$, otherwise, $\mu_{k}=0$. Hence, the scheduled user set is given by $\mathcal{S}=\left\{k|\mu_{k}=1, \forall k\in\mathcal{K}\right\}$ and problem~\eqref{CUSBD03} is rewritten as\footnote{The authors of~\cite{JSACHong2013} focus on investigating the problem of user-centric BS clustering with aiming to reduce the coordination overhead in joint transmission networks. So that each user is served by only a small number of (potentially overlapping) BSs. Consequently, the algorithm designed in~\cite{JSACHong2013} can only give which users the BS sends data to, but it cannot directly tell whether a BS is closed or not. This also means that their problem formulation cannot be used to describe the joint user scheduling and beamforming design problem considered in this work.}
\begin{subequations}\label{CUSBD05}
\begin{align}
&\max_{\left\{\mu_{k},p_{k},\mathbf{w}_{k}\right\}} \sum\limits_{k\in\mathcal{K}}\mu_{k}\overrightarrow{R}_{k}, \label{CUSBD05a}\\
\mathrm{s.t.}~&\mu_{k}\in\left\{0,1\right\}, \forall k\in\mathcal{K},\label{CUSBD05b}\\
&p_{k}\geq 0, \left\|\mathbf{w}_{k}\right\|_{2}=1, \forall k\in\mathcal{K},\label{CUSBD05c}\\
&\mu_{k}r_{k}\leq \overrightarrow{R}_{k}, \forall k\in\mathcal{K},\sum\limits_{k\in\mathcal{K}}\mu_{k}\leq BN_{\mathrm{t}},\label{CUSBD05d}\\
&\sum\limits_{k\in\mathcal{K}}p_{k}\left\|\mathbf{Q}_{b}\mathbf{w}_{k}\right\|_{2}^{2}\leq P_{b}, \forall b\in\mathcal{B},\label{CUSBD05e}
\end{align}
\end{subequations}
where $\overrightarrow{R}_{k}$ indicates the downlink rate of the $k$-th user, which is calculated as $\overrightarrow{R}_{k}=\log_{2}\left(1+\overrightarrow{\gamma}_{k}\right)$ with $\overrightarrow{\gamma}_{k}$ being
\begin{equation}\label{CUSBD06}
\overrightarrow{\gamma}_{k}=\frac{p_{k}\left|\overline{\mathbf{h}}_{k}^{H}\mathbf{w}_{k}\right|^{2}}
{\sum\limits_{l\neq k,l\in\mathcal{K}}p_{l}\left|\overline{\mathbf{h}}_{k}^{H}\mathbf{w}_{l}\right|^{2}+1}.
\end{equation}
Accordingly, using the formulation of problem~\eqref{CUSBD05}, we have the following conclusion.
\begin{theorem}\label{CUSBDT01}
In problem~\eqref{CUSBD05}, if $\mu_{k}=1$, then $p_{k}>0$, $k\in\mathcal{K}$, otherwise, $p_{k}=0$, $k\in\mathcal{K}$.
\end{theorem}
\begin{IEEEproof}
It is easy to see that if binary variable $\mu_{k}=1$, to satisfy the required minimum user rate of the $k$-th user scheduled, transmitting power $p_{k}$ used at the BSs for the $k$-th user must be larger than zero.

Let $\left\{p_{k}^{'}\right\}$ and $\left\{\mathbf{w}_{k}^{'}\right\}$ be the optimal solution to problem~\eqref{CUSBD05}, and the corresponding binary variables $\bm{u}^{'}$. Assume that binary variable $\mu_{m}^{'}=0$, i.e., $m\in\mathcal{K}\backslash\mathcal{S}$, but $p_{m}^{'}\neq 0$, while $\mu_{j}^{'}=0$ and $p_{j}^{'}=0$,$\forall j\in\mathcal{K}\backslash\{\mathcal{S}\bigcup\{m\}\}$. Thus, we have the following relation:
\begin{equation}\label{CUSBD07}
\sum\limits_{k\in\mathcal{S}}p_{k}^{'}\left\|\mathbf{Q}_{b}\mathbf{w}_{k}^{'}\right\|_{2}^{2}< \sum\limits_{k\in\mathcal{S}}p_{k}^{'}\left\|\mathbf{Q}_{b}\mathbf{w}_{k}^{'}\right\|_{2}^{2}
+p_{m}^{'}\left\|\mathbf{Q}_{b}\mathbf{w}_{m}^{'}\right\|_{2}^{2}\leq P_{b}, \forall b\in\mathcal{B}.
\end{equation}
This implies that we can redistribute equal power $p_{l}^{'}$ among the set $\mathcal{S}$ of users scheduled. In other words, let $\left\{p_{k}^{''}=\alpha p_{k}^{'}\right\}$ and $\left\{\mathbf{w}_{k}^{''}=\mathbf{w}_{k}^{'}\right\}$, $\forall k\in\mathcal{S}$, $\alpha>1$, and $p_{m}^{''}=0$, such that
\begin{equation}\label{CUSBD08}
\sum\limits_{k\in\mathcal{S}}p_{k}^{''}\left\|\mathbf{Q}_{b}\mathbf{w}_{k}^{''}\right\|_{2}^{2}\leq\sum\limits_{k\in\mathcal{S}}p_{k}^{'}\left\|\mathbf{Q}_{b}\mathbf{w}_{k}^{'}\right\|_{2}^{2}
+p_{m}^{'}\left\|\mathbf{Q}_{b}\mathbf{w}_{m}^{'}\right\|_{2}^{2}\leq P_{b}, \forall b\in\mathcal{B}.
\end{equation}
Then, we have
\begin{equation}\label{CUSBD09}
\frac{p_{k}^{''}\left|\overline{\mathbf{h}}_{k}^{H}\mathbf{w}_{k}^{'}\right|^{2}}
{\sum\limits_{l\neq k,l\in\mathcal{S}}p_{l}^{''}\left|\overline{\mathbf{h}}_{k}^{H}\mathbf{w}_{l}^{'}\right|^{2}+1}
=\frac{p_{k}^{'}\left|\overline{\mathbf{h}}_{k}^{H}\mathbf{w}_{k}^{'}\right|^{2}}
{\sum\limits_{l\neq k,l\in\mathcal{S}}p_{l}^{'}\left|\overline{\mathbf{h}}_{k}^{H}\mathbf{w}_{l}^{'}\right|^{2}+\frac{1}{\alpha}}
>\frac{p_{k}^{'}\left|\overline{\mathbf{h}}_{k}^{H}\mathbf{w}_{k}^{'}\right|^{2}}
{\sum\limits_{l\neq k,l\in\mathcal{S}}p_{l}^{'}\left|\overline{\mathbf{h}}_{k}^{H}\mathbf{w}_{l}^{'}\right|^{2}+p_{m}^{'}\left|\overline{\mathbf{h}}_{k}^{H}\mathbf{w}_{m}^{'}\right|^{2}+1}.
\end{equation}
This means that it would be possible to obtain a larger objective function with $\left\{p_{k}^{''}\right\}$ and $\left\{\mathbf{w}_{k}^{''}\right\}$ than the global maximum with $\left\{p_{k}^{'}\right\}$ and $\left\{\mathbf{w}_{k}^{'}\right\}$, which induces a contradiction. The conclusions are proved.
\end{IEEEproof}
\begin{remark}\label{Remark02}
Using Theorem~\ref{CUSBDT01}, if binary variable $\mu_{k}= 1$, then the transmitting power $p_{k}$ used at the BSs for the $k$-th user must be larger than zero, which is to satisfy the demand of non-zero minimum user rate of user scheduled. Contrary, if binary variable $\mu_{k}= 0$, then we have the transmitting power $p_{k}=0$ used at the BSs for the $k$-th user,$\forall k\in\mathcal{K}$.
\end{remark}

\begin{proposition}\label{CUSBDP01}
Problem~\eqref{CUSBD05} is equivalent to problem~\eqref{CUSBD03}.
\end{proposition}
\begin{IEEEproof}
The conclusion is obvious according to the conclusion obtained in Theorem~\ref{CUSBDT01}.
\end{IEEEproof}

According to the conclusion in Proposition~\ref{CUSBDP01}, in the sequel, we focus on addressing problem~\eqref{CUSBD05} instead of problem~\eqref{CUSBD03}. In problem~\eqref{CUSBD05}, the uncertain set $\mathcal{S}$ is replaced with a series of binary variables $\mu_{k}$, $\forall k\in\mathcal{K}$. However, problem~\eqref{CUSBD05} is a mixed-integer programming problem, which is still difficult to obtain its global optimum. In what follows, we resort to the uplink-downlink duality theory to address problem~\eqref{CUSBD05}.

\subsection{Dual description of problem~\eqref{CUSBD05}}

The research in~\cite{TITZhang2012} showed that the complex downlink resource allocation problem can be transformed into a relative tractable uplink optimization problem for wireless communication systems. This is because there is an analytical solution of the beamforming vectors of the uplink optimization problem. Motivated by this observation, in the sequel, we focus on obtaining a dual description of problem~\eqref{CUSBD05}, such that the solution to problem~\eqref{CUSBD05} can be obtained. Introducing auxiliary variables $\lambda_{b}\geq 0$, $\forall b\in\mathcal{B}$ with $\sum\limits_{b\in\mathcal{B}}\lambda_{b}\neq 0$, we transform problem~\eqref{CUSBD05} into the following form
\begin{subequations}\label{CUSBD10}
\begin{align}
&\psi\left(\bm{\lambda}\right)=\max_{\left\{\mu_{k},p_{k},\mathbf{w}_{k}\right\}} \sum\limits_{k\in\mathcal{K}}\mu_{k}\overrightarrow{R}_{k}, \label{CUSBD10a}\\
\mathrm{s.t.}~&\mu_{k}\in\left\{0,1\right\}, \forall k\in\mathcal{K},\label{CUSBD10b}\\
&p_{k}\geq 0, \left\|\mathbf{w}_{k}\right\|_{2}=1, \forall k\in\mathcal{K},\label{CUSBD10c}\\
&\mu_{k}r_{k}\leq \overrightarrow{R}_{k}, \forall k\in\mathcal{K},\sum\limits_{k\in\mathcal{K}}\mu_{k}\leq BN_{\mathrm{t}},\label{CUSBD10d}\\
&\sum\limits_{b\in\mathcal{B}}\sum\limits_{k\in\mathcal{K}}\lambda_{b}p_{k}\left\|\mathbf{Q}_{b}\mathbf{w}_{k}\right\|_{2}^{2}\leq \sum\limits_{b\in\mathcal{B}}\lambda_{b}P_{b}.\label{CUSBD10e}
\end{align}
\end{subequations}
where $\bm{\lambda}=\{\lambda_{1},\lambda_{2},\cdots,\lambda_{B}\}$. Note that problem~\eqref{CUSBD10} is the problem of coordinated user scheduling and beamforming design for joint transmission networks under the total power constraint. It is not difficult to find that a feasible solution to problem~\eqref{CUSBD05} is also a feasible solution to problem~\eqref{CUSBD10}~\cite{TITZhang2012}. Accordingly, the optimal solution to problem~\eqref{CUSBD05} can be obtained via solving problem~$\min\limits_{\bm{\lambda}}\psi\left(\bm{\lambda}\right)$~\cite{TCOMHe2015}. Solving problem~$\min\limits_{\bm{\lambda}}\psi\left(\bm{\lambda}\right)$ can be achieved via the sub-gradient method, which guarantees to obtain the optimal solution~\cite{BookBorwein2006}. The sub-gradient of function $\psi\left(\bm{\lambda}\right)$ is given by the following Proposition.
\begin{proposition}\label{CUSBDP02}
The sub-gradient $\mathbf{g}=\left[g_{1},g_{2},\cdots,g_{B}\right]$ of function $\psi\left(\bm{\lambda}\right)$ with a fixed parameter $\bm{\lambda}$ is given by $g_{b}=P_{b}-\sum\limits_{k\in\mathcal{K}}p_{k}\left\|\mathbf{Q}_{b}\mathbf{w}_{k}\right\|_{2}^{2}, \forall b\in\mathcal{B}$.
\end{proposition}

In what follows, we pay our attention to solving problem~\eqref{CUSBD10} with fixed $\bm{\lambda}$. Given $\bm{\lambda}$, there are still two challenges, i.e., the non-convexity of the $0$-$1$ constraint and the user rate $\overrightarrow{R}_{k}$. We begin with obtaining a virtual unplink dual form of problem~\eqref{CUSBD10}, which is described in the following Theorem.
\begin{theorem}\label{CUSBDT02}
Given $\bm{\lambda}$, the downlink optimization problem~\eqref{CUSBD10} is dual to the following virtual uplink optimization problem~\eqref{CUSBD11},  i.e.,
\begin{subequations}\label{CUSBD11}
\begin{align}
&\max_{\left\{\mu_{k},q_{k},\mathbf{w}_{k}\right\}} \sum\limits_{k\in\mathcal{K}}\mu_{k}\overleftarrow{R}_{k}, \label{CUSBD11a}\\
\mathrm{s.t.}~&\mu_{k}\in\left\{0,1\right\}, \forall k\in\mathcal{K},\label{CUSBD11b}\\
&q_{k}\geq 0, \left\|\mathbf{w}_{k}\right\|_{2}=1, \forall k\in\mathcal{K},\label{CUSBD11c}\\
&\mu_{k}r_{k}\leq \overleftarrow{R}_{k}, \forall k\in\mathcal{K},\sum\limits_{k\in\mathcal{K}}\mu_{k}\leq BN_{\mathrm{t}},\label{CUSBD11d}\\
&\sum\limits_{k\in\mathcal{K}}q_{k}\leq \sum\limits_{b\in\mathcal{B}}\lambda_{b}P_{b}.\label{CUSBD11e}
\end{align}
\end{subequations}
where $q_{k}$ is the virtual uplink transmit power of the $k$-th user and $\overleftarrow{R}_{k}$ is calculated as $\overleftarrow{R}_{k}=\log_{2}\left(1+\overleftarrow{\gamma}_{k}\right)$ with $\overleftarrow{\gamma}_{k}$ being the virtual uplink SINR of the $k$-th user, calculated as
\begin{equation}\label{CUSBD12}
\overleftarrow{\gamma}_{k}=\frac{q_{k}\left|\overline{\mathbf{h}}_{k}^{H}\mathbf{w}_{k}\right|^{2}}
{\sum\limits_{l\neq k,l\in\mathcal{K}}q_{l}\left|\overline{\mathbf{h}}_{l}^{H}\mathbf{w}_{k}\right|^{2}+\left\|\mathbf{Q}\mathbf{w}_{k}\right\|_{2}^{2}}.
\end{equation}
where $\mathbf{Q}=\sum\limits_{b\in\mathcal{B}}\sqrt{\lambda_{b}}\mathbf{Q}_{b}$.
\end{theorem}
\begin{IEEEproof}
To proof the duality of problem~\eqref{CUSBD10} and problem~\eqref{CUSBD11}, we try to prove that they have the same SINR feasible region. Let $q_{k}$ be the virtual uplink transmitting power of the $k$-th user, and the virtual uplink SINR of the $k$-th user is given by~\eqref{CUSBD12}. Let $\overrightarrow{\gamma}_{k}=\overleftarrow{\gamma}_{k}$, then we have
\begin{equation}\label{CUSBD13}
\sum\limits_{l\neq k,l\in\mathcal{K}}q_{k}p_{l}\left|\overline{\mathbf{h}}_{k}^{H}\mathbf{w}_{l}\right|^{2}+q_{k}
=\sum\limits_{l\neq k,l\in\mathcal{K}}p_{k}q_{l}\left|\overline{\mathbf{h}}_{l}^{H}\mathbf{w}_{k}\right|^{2}+p_{k}\left\|\mathbf{Q}\mathbf{w}_{k}\right\|_{2}^{2}
\end{equation}
Add the item $q_{k}p_{k}\left|\overline{\mathbf{h}}_{k}^{H}\mathbf{w}_{k}\right|^{2}$ on both sides of~\eqref{CUSBD13} and summing with respect to the subscript $k$, we have
\begin{equation}\label{CUSBD14}
\sum\limits_{k\in\mathcal{K}}q_{k}
=\sum\limits_{k\in\mathcal{K}}p_{k}\left\|\mathbf{Q}\mathbf{w}_{k}\right\|_{2}^{2}
\end{equation}
Recalling $\mathbf{Q}=\sum\limits_{b\in\mathcal{B}}\sqrt{\lambda_{b}}\mathbf{Q}_{b}$, we further have
\begin{equation}\label{CUSBD15}
\sum\limits_{k\in\mathcal{K}}q_{k}
=\sum\limits_{b\in\mathcal{B}}\sum\limits_{k\in\mathcal{K}}\lambda_{b}p_{k}\left\|\mathbf{Q}_{b}\mathbf{w}_{k}\right\|_{2}^{2}
\end{equation}
This implies that under the same total transmit power constraint, the downlink and the virtual uplink have the same SINR region. Therefore, problem~\eqref{CUSBD10} and problem~\eqref{CUSBD11} are dual problems.
\end{IEEEproof}

\begin{remark}\label{Remark03}
The conclusion given in Theorem~\ref{CUSBDT02} is an extension of Theorem~1, which gives out a dual description of the power minimization problem of the downlink multiuser system with per-antenna power constraint in~\cite{TSPYu2007}. However, in this paper, we obtain the dual description of the sum rate maximization problem of the downlink of multicell joint transmission systems with per-BS power constraint.
\end{remark}
\begin{proposition}\label{CUSBDP03}
In problem~\eqref{CUSBD11}, if $\mu_{k}=1$, then $q_{k}>0$, otherwise, $q_{k}=0$, $k\in\mathcal{K}$.
\end{proposition}
\begin{IEEEproof}
The proof is omitted.
\end{IEEEproof}

The other thing is that once the optimal beamforming vectors $\left\{\mathbf{w}_{k}^{\left(*\right)}\right\}$ and virtual uplink power allocation $q_{k}^{\left(*\right)}$ are given, how to calculate the  power allocation $\left\{p_{k}^{\left(*\right)}\right\}$ needs to be addressed, or vice versa, $k\in\mathcal{S}$. For simplicity, let $\mathcal{S}=\left\{1',\cdots,k',\cdots,K'\right\}\subseteq\mathcal{K}$ be the index set of user scheduled\footnote{The mapping the set $\mathcal{S}$ and set $\mathcal{K}$ can be determined according to the scheduling results. For example, $\mathcal{S}=\left\{2,3,4\right\}\subseteq\mathcal{K}=\left\{1,2,3,4,5\right\}$, i.e., $1'=2$, $2'=3$, and $3'=4$.}, $K'\leq K$. The following Proposition is used to address this problem~\cite{TCOMHe2015}.
\begin{proposition}\label{CUSBDP04}
For fixed auxiliary variables $\bm{\lambda}$, let $\left\{\mathbf{w}_{k'}^{\left(*\right)}\right\}$ and $\left\{q_{k'}^{\left(*\right)}\right\}$ be the optimal beamforming vectors and power allocation of the virtual uplink optimization problem~\eqref{CUSBD11}, the optimal power allocation $\left\{p_{k'}^{\left(*\right)}\right\}$ of the downlink optimization problem~\eqref{CUSBD10} is given by
\begin{equation}\label{CUSBD16}
\widetilde{\mathbf{p}}^{\left(*\right)}=\widetilde{\mathbf{Q}}^{\left(*\right)}\widetilde{\mathbf{p}}^{\left(*\right)}, \text{with}, \widetilde{\mathbf{p}}_{K'+1}^{\left(*\right)}=1,
\end{equation}
where $\widetilde{\mathbf{p}}^{\left(*\right)}=\left[\mathbf{p}^{\left(*\right)}\atop 1\right]$, $\mathbf{p}^{\left(*\right)}=\left[p_{1'}^{\left(*\right)},p_{2'}^{\left(*\right)},\cdots,p_{K'}^{\left(*\right)}\right]^{T}$, matrix $\widetilde{\mathbf{Q}}$ is defined as
\begin{equation}\label{CUSBD17}
\widetilde{\mathbf{Q}}^{\left(*\right)}={
\left[
\begin{array}{ccc}
\mathbf{D}^{\left(*\right)}\mathbf{\Psi}^{\left(*\right)} & \mathbf{D}^{\left(*\right)}\mathbf{1}\\
\frac{1}{\sum\limits_{k\in\mathcal{S}}q_{k'}}\bm{\nu}^{T}\mathbf{D}^{\left(*\right)}\mathbf{\Psi}^{\left(*\right)} & \frac{1}{\sum\limits_{k\in\mathcal{S}}q_{k'}}\bm{\nu}^{T}\mathbf{D}^{\left(*\right)}\mathbf{1},
\end{array}
\right ]}
\end{equation}
where $\mathbf{1}$ is a $K'$-dimension all-one vector, matrices $\mathbf{D}^{\left(*\right)}$ and $\mathbf{\Psi}^{\left(*\right)}$ are respectively calculated as
\begin{equation}\label{CUSBD18}
\left[\mathbf{D}^{\left(*\right)}\right]_{k',l'}=\left\{
\begin{aligned}
\frac{\overleftarrow{\gamma}_{k'}^{\left(*\right)}}{\left|\overline{\mathbf{h}}_{k'}^{H}\mathbf{w}_{k'}^{\left(*\right)}\right|^{2}}, k'=l', \\
0\quad\quad,k'\neq l'.
\end{aligned}
\right.
\end{equation}
\begin{equation}\label{CUSBD19}
\left[\mathbf{\Psi}^{\left(*\right)}\right]_{k',l'}=\left\{
\begin{aligned}
 0\quad\quad ,k'=l', \\
\left|\overline{\mathbf{h}}_{k'}^{H}\mathbf{w}_{l'}^{\left(*\right)}\right|^{2}, k'\neq l'.
\end{aligned}
\right.
\end{equation}
where $\overleftarrow{\gamma}_{k'}^{\left(*\right)}$ is the corresponding virtual uplink optimal SINR of the $k'$-th user.
\end{proposition}
\begin{proof}
Combining $\overrightarrow{\gamma}_{k'}=\overleftarrow{\gamma}_{k'}$ with~\eqref{CUSBD14}, we have
\begin{equation}\label{CUSBD20}
\left\{
\begin{aligned}
&\overrightarrow{\gamma}_{k'}=\overleftarrow{\gamma}_{k'}, \quad \forall k'\in\mathcal{S},\\
&\sum\limits_{k'\in\mathcal{S}}p_{k'}\left\|\mathbf{Q}\mathbf{w}_{k'}\right\|_{2}^{2}=\sum\limits_{k\in\mathcal{S}}q_{k'}.
\end{aligned}
\right.
\end{equation}
After some basic operation, we have
\begin{subnumcases} {\label{CUSBD21} }
\mathbf{p}=\mathbf{D}\mathbf{\Psi}\mathbf{p}+\mathbf{D}\mathbf{1}\label{CUSBD21a}\\
\sum\limits_{k'\in\mathcal{S}}p_{k'}\left\|\mathbf{Q}\mathbf{w}_{k'}\right\|_{2}^{2}=\sum\limits_{k\in\mathcal{S}}q_{k'}.\label{CUSBD21b}
\end{subnumcases}
%
where $\mathbf{p}=\left[p_{1'},p_{2'},\cdots,p_{K'}\right]^{T}$, $\mathbf{1}$ is a $K'$-dimension all-one vector, and matrix $\mathbf{\Psi}$ is given by
\begin{equation}\label{CUSBD22}
\left[\mathbf{D}\right]_{k',l'}=\left\{
\begin{aligned}
\frac{\overleftarrow{\gamma}_{k'}}{\left|\overline{\mathbf{h}}_{k'}^{H}\mathbf{w}_{k'}\right|^{2}}, k'=l', \\
0\quad\quad,k'\neq l'.
\end{aligned}
\right.
\end{equation}
\begin{equation}\label{CUSBD23}
\left[\mathbf{\Psi}\right]_{k',l'}=\left\{
\begin{aligned}
 0\quad\quad ,k'=l', \\
\left|\overline{\mathbf{h}}_{k'}^{H}\mathbf{w}_{l'}\right|^{2}, k'\neq l'.
\end{aligned}
\right.
\end{equation}
Let $\bm{\nu}=\left[\left\|\mathbf{Q}\mathbf{w}_{1'}\right\|_{2}^{2},\left\|\mathbf{Q}\mathbf{w}_{2'}\right\|_{2}^{2},\cdots,\left\|\mathbf{Q}\mathbf{w}_{K'}\right\|_{2}^{2}\right]^{T}$ and using~\eqref{CUSBD21a}, we have
\begin{equation}\label{CUSBD24}
1=\frac{1}{\sum\limits_{k\in\mathcal{S}}q_{k'}}\bm{\nu}^{T}\mathbf{D}\mathbf{\Psi}\mathbf{p}+\frac{1}{\sum\limits_{k\in\mathcal{S}}q_{k'}}\bm{\nu}^{T}\mathbf{D}\mathbf{1}
\end{equation}
Defining an extended vector $\widetilde{\mathbf{p}}=\left[\mathbf{p}\atop 1\right]$ and an extended matrix $\widetilde{\mathbf{Q}}$, i.e.,
\begin{equation}\label{CUSBD25}
\widetilde{\mathbf{Q}}={
\left[
\begin{array}{ccc}
\mathbf{D}\mathbf{\Psi} & \mathbf{D}\mathbf{1}\\
\frac{1}{\sum\limits_{k\in\mathcal{S}}q_{k'}}\bm{\nu}^{T}\mathbf{D}\mathbf{\Psi} & \frac{1}{\sum\limits_{k\in\mathcal{S}}q_{k'}}\bm{\nu}^{T}\mathbf{D}\mathbf{1}
\end{array}
\right ]}
\end{equation}
Thus, we have the following eigensystem by combining~\eqref{CUSBD21a} and~\eqref{CUSBD25}:
\begin{equation}\label{CUSBD26}
\widetilde{\mathbf{p}}=\widetilde{\mathbf{Q}}\widetilde{\mathbf{p}}, \text{with}, \widetilde{\mathbf{p}}_{K'+1}=1.
\end{equation}
It is not difficult to find that the optimal power vector $\mathbf{p}$ is obtained as the first $K'$ components of the dominant eigenvector of $\widetilde{\mathbf{Q}}$, which can be scaled so that its last component equals one~\cite{TVTSchubert2004}.
\end{proof}
\begin{proposition}\label{CUSBDP05}
For fixed auxiliary variables $\bm{\lambda}$, let $\left\{\mathbf{w}_{k'}^{\left(*\right)}\right\}$ and $\left\{p_{k'}^{\left(*\right)}\right\}$ be the optimal beamforming vectors and power allocation of the downlink optimization problem~\eqref{CUSBD10}, the optimal power allocation $\left\{q_{k'}^{\left(*\right)}\right\}$ of the virtual uplink optimization problem~\eqref{CUSBD11} is given by
\begin{equation}\label{CUSBD27}
\widetilde{\mathbf{q}}^{\left(*\right)}=\widehat{\mathbf{Q}}^{\left(*\right)}\widetilde{\mathbf{q}}^{\left(*\right)}, \text{with}, \widetilde{\mathbf{q}}_{K'+1}^{\left(*\right)}=1,
\end{equation}
where $\widetilde{\mathbf{q}}^{\left(*\right)}=\left[\mathbf{q}^{\left(*\right)}\atop 1\right]$, $\mathbf{q}^{\left(*\right)}=\left[q_{1'}^{\left(*\right)},q_{2'}^{\left(*\right)},\cdots,q_{K'}^{\left(*\right)}\right]^{T}$, matrix $\widehat{\mathbf{Q}}$ is defined as
\begin{equation}\label{CUSBD28}
\widehat{\mathbf{Q}}^{\left(*\right)}={
\left[
\begin{array}{ccc}
\mathbf{O}^{\left(*\right)}\mathbf{\Phi}^{\left(*\right)} & \mathbf{O}^{\left(*\right)}\mathbf{n}\\
\frac{1}{\sum\limits_{k'\in\mathcal{S}}p_{k'}\left\|\mathbf{Q}\mathbf{w}_{k'}\right\|_{2}^{2}}\mathbf{1}^{T}\mathbf{O}^{\left(*\right)}\mathbf{\Phi}^{\left(*\right)} & \frac{1}{\sum\limits_{k'\in\mathcal{S}}p_{k'}\left\|\mathbf{Q}\mathbf{w}_{k'}\right\|_{2}^{2}}\mathbf{1}^{T}\mathbf{O}^{\left(*\right)}\mathbf{n},
\end{array}
\right ]},
\end{equation}
where $\mathbf{n}=\left[\left\|\mathbf{Q}\mathbf{w}_{1'}^{\left(*\right)}\right\|_{2}^{2},
\left\|\mathbf{Q}\mathbf{w}_{2'}^{\left(*\right)}\right\|_{2}^{2},\cdots,\left\|\mathbf{Q}\mathbf{w}_{K'}^{\left(*\right)}\right\|_{2}^{2}\right]^{T}$, matrices $\mathbf{O}^{\left(*\right)}$ and $\mathbf{\Phi}^{\left(*\right)}$ is calculated as
\begin{equation}\label{CUSBD29}
\left[\mathbf{O}^{\left(*\right)}\right]_{k',l'}=\left\{
\begin{aligned}
\frac{\overrightarrow{\gamma}_{k'}^{\left(*\right)}}{\left|\overline{\mathbf{h}}_{k'}^{H}\mathbf{w}_{k'}^{\left(*\right)}\right|^{2}}, k'=l', \\
0\quad\quad,k'\neq l'.
\end{aligned}
\right.
\end{equation}
\begin{equation}\label{CUSBD30}
\left[\mathbf{\Phi}^{\left(*\right)}\right]_{k',l'}=\left\{
\begin{aligned}
0\quad\quad, k'=l', \\
\left|\overline{\mathbf{h}}_{l'}^{H}\mathbf{w}_{k'}^{\left(*\right)}\right|^{2}, k'\neq l'.
\end{aligned}
\right.
\end{equation}
\end{proposition}
\begin{proof}
The proof of this Lemma is similar with that of Proposition~\ref{CUSBDP04}. The details are omitted.
\end{proof}

In this section, we obtain a tractable form, i.e., problem~\eqref{CUSBD11}, of the original optimization problem~\eqref{CUSBD03} via exploiting the uplink-downlink duality. But, due to the existing of the coupling variables $\left\{q_{k},\mathbf{w}_{k}\right\}$ and the mixed 0-1 integer and continue variables programming, it is till hard to address problem~\eqref{CUSBD11}. In the sequel,  to solve this difficult problem, an alternative optimization algorithm is proposed.

\section{\label{AlgorithmDesign} Design of Optimization Algorithm}

In this section, we focus on designing an alternative optimization algorithm to address problem~\eqref{CUSBD11}. Furthermore, to compare the performance of the proposed algorithm, we exploit the brute-force search and a heuristic user scheduling algorithm with zero-forcing beamforming (ZFBF) to address problem~\eqref{CUSBD11}.
\subsection{Alterative Optimization Algorithm}
Compared to the original problem~\eqref{CUSBD03}, in problem~\eqref{CUSBD11}, the combinational optimization problem has been transformed into a mixed 0-1 and continue variables programming problem, which has a slightly traceable form. However, the coupling of $\left\{q_{k},\mathbf{w}_{k}\right\}$ makes problem~\eqref{CUSBD11} challenging. In what follows, we resort to design an alternative algorithm to release the coupling of $\left\{q_{k},\mathbf{w}_{k}\right\}$. Specifically, we firstly optimize the beamforming vector $\left\{\mathbf{w}_{k}\right\}$ by fixing uplink transmitting powers $\left\{q_{k}\right\}$, and the optimal beamforming vectors $\left\{\mathbf{w}_{k}\right\}$ are the minimum mean square error receiver with fixed $\left\{q_{k}\right\}$, i.e.,~\cite{TWCDah2010}
\begin{equation}\label{CUSBD31}
\mathbf{w}_{k}^{*}=\frac{\left(\mathbf{Q}+\sum\limits_{l\in\mathcal{K}}q_{l}\overline{\mathbf{h}}_{l}\overline{\mathbf{h}}_{l}^{H}\right)^{-1}\overline{\mathbf{h}}_{k}}
{\left\|\left(\mathbf{Q}+\sum\limits_{l\in\mathcal{K}}q_{l}\overline{\mathbf{h}}_{l}\overline{\mathbf{h}}_{l}^{H}\right)^{-1}\overline{\mathbf{h}}_{k}\right\|}.
\end{equation}

Then, in what follows, we focus on addressing problem~\eqref{CUSBD11} with respect to variables $\left\{\mu_{k},q_{k}\right\}$ with fixed beamforming vectors $\left\{\mathbf{w}_{k}\right\}$ and auxiliary variable $\left\{\lambda_{b}\right\}$, $\forall k\in\mathcal{K}$ and $\forall b\in\mathcal{B}$. Given beamforming vectors $\left\{\mathbf{w}_{k}\right\}$ and auxiliary variable $\left\{\lambda_{b}\right\}$, $\forall k\in\mathcal{K}$ and $\forall b\in\mathcal{B}$, problem~\eqref{CUSBD11} is rewritten as follows
\begin{subequations}\label{CUSBD32}
\begin{align}
&\max_{\left\{\mu_{k},q_{k}\right\}} \sum\limits_{k\in\mathcal{K}}\mu_{k}\log_{2}\left(1+\overleftarrow{\gamma}_{k}\right), \label{CUSBD32a}\\
\mathrm{s.t.}~&\mu_{k}\in\left\{0,1\right\}, \forall k\in\mathcal{K},\label{CUSBD32b}\\
&\mu_{k}r_{k}\leq \overleftarrow{R}_{k}, \forall k\in\mathcal{K},\sum\limits_{k\in\mathcal{K}}\mu_{k}\leq BN_{\mathrm{t}},\label{CUSBD32c}\\
&q_{k}\geq 0, \forall k\in\mathcal{K}, \sum\limits_{k\in\mathcal{K}}q_{k}\leq \sum\limits_{b\in\mathcal{B}}\lambda_{b}P_{b}.\label{CUSBD32d}
\end{align}
\end{subequations}
Note that problem~\eqref{CUSBD32} is a mixed integer programming problem with non-convex objective function. Introduce auxiliary variables $\theta_{k}$, $\vartheta_{k}$, $\kappa_{k}$, and exploiting the feature of $\mu_{k}\in\left\{0,1\right\}$, $\forall k\in\mathcal{K}$, we transform problem~\eqref{CUSBD32} into continued non-convex optimization problem, i.e.,
\begin{subequations}\label{CUSBD33}
\begin{align}
&\max_{\left\{\mu_{k},q_{k},\theta_{k},\kappa_{k},\vartheta_{k}\right\}} \sum\limits_{k\in\mathcal{K}}\kappa_{k}^{2}, \label{CUSBD33a}\\
\mathrm{s.t.}~&\sum\limits_{k\in\mathcal{K}}\left(\mu_{k}-\mu_{k}^{2}\right)\leq 0,\label{CUSBD33b}\\
&\psi_{k}\left(\theta_{k},\mathbf{q}\right)-\phi_{k}\left(\theta_{k},\mathbf{q}\right)\leq 0, \forall k\in\mathcal{K},\label{CUSBD33c}\\
&\tilde{\psi}_{k}\left(\mu_{k},\mathbf{q}\right)-\tilde{\phi}_{k}\left(\mu_{k},\mathbf{q}\right)\leq 0, \forall k\in\mathcal{K},\label{CUSBD33d}\\
&\vartheta_{k}\leq\log_{2}\left(1+\theta_{k}\right), \forall k\in\mathcal{K},\label{CUSBD33e}\\
&\kappa_{k}^{2}\leq\mu_{k}\vartheta_{k}, \forall k\in\mathcal{K},\label{CUSBD33f}\\
&0\leq\mu_{k}\leq 1, \forall k\in\mathcal{K},\label{CUSBD33g}\\
&q_{k}\geq 0, \forall k\in\mathcal{K}, \sum\limits_{k\in\mathcal{K}}q_{k}\leq \sum\limits_{b\in\mathcal{B}}\lambda_{b}P_{b},\label{CUSBD33h}\\
& \sum\limits_{k\in\mathcal{K}}\mu_{k}\leq BN_{\mathrm{t}},\label{CUSBD33i}
\end{align}
\end{subequations}
where $\widetilde{\gamma}_{k}=2^{r_{k}}-1$, $\psi_{k}\left(\theta_{k},\mathbf{q}\right)$ and $\phi_{k}\left(\theta_{k},\mathbf{q}\right)$ are defined as
\begin{subequations}\label{CUSBD34}
\begin{align}
\psi_{k}\left(\theta_{k},\mathbf{q}\right)&=\left\|\mathbf{Q}\mathbf{w}_{k}\right\|_{2}^{2}\theta_{k}-q_{k}\left|\overline{\mathbf{h}}_{k}^{H}\mathbf{w}_{k}\right|^{2}
+\varphi_{k}\left(\theta_{k},\mathbf{q}\right),\label{CUSBD34a}\\
\varphi_{k}\left(\theta_{k},\mathbf{q}\right)&\triangleq\frac{1}{2}\left(\theta_{k}+\sum\limits_{l\neq k,l\in\mathcal{K}}q_{l}\left|\overline{\mathbf{h}}_{l}^{H}\mathbf{w}_{k}\right|^{2}\right)^{2},\label{CUSBD34b}\\
\phi_{k}\left(\theta_{k},\mathbf{q}\right)&\triangleq\frac{1}{2}\theta_{k}^{2}+\frac{1}{2}\left(\sum\limits_{l\neq k,l\in\mathcal{K}}q_{l}\left|\overline{\mathbf{h}}_{l}^{H}\mathbf{w}_{k}\right|^{2}\right)^{2},\label{CUSBD34c}
\end{align}
\end{subequations}
$\tilde{\psi}_{k}\left(\mu_{k},\mathbf{q}\right)$ and $\tilde{\phi}_{k}\left(\mu_{k},\mathbf{q}\right)$ are defined as
\begin{subequations}\label{CUSBD35}
\begin{align}
\tilde{\psi}_{k}\left(\mu_{k},\mathbf{q}\right)&=\widetilde{\gamma}_{k}\left\|\mathbf{Q}\mathbf{w}_{k}\right\|_{2}^{2}\mu_{k}-q_{k}\left|\overline{\mathbf{h}}_{k}^{H}\mathbf{w}_{k}\right|^{2}
+\tilde{\varphi}_{k}\left(\mu_{k},\mathbf{q}\right),\label{CUSBD35a}\\
\tilde{\varphi}_{k}\left(\mu_{k},\mathbf{q}\right)&\triangleq\frac{1}{2}\left(\widetilde{\gamma}_{k}\mu_{k}+\sum\limits_{l\neq k,l\in\mathcal{K}}q_{l}\left|\overline{\mathbf{h}}_{l}^{H}\mathbf{w}_{k}\right|^{2}\right)^{2},\label{CUSBD35b}\\
\tilde{\phi}_{k}\left(\mu_{k},\mathbf{q}\right)&\triangleq\frac{1}{2}\widetilde{\gamma}_{k}^{2}\mu_{k}^{2}+\frac{1}{2}\left(\sum\limits_{l\neq k,l\in\mathcal{K}}q_{l}\left|\overline{\mathbf{h}}_{l}^{H}\mathbf{w}_{k}\right|^{2}\right)^{2}.\label{CUSBD35c}
\end{align}
\end{subequations}
At the optimal point of problem~\eqref{CUSBD33}, constraints~\eqref{CUSBD33c},~\eqref{CUSBD33e}, and~\eqref{CUSBD33f} are activated. Introduce a turnable control parameter $\tau>0$, we move constraint~\eqref{CUSBD33b} into the objective function, i.e.,~\cite{Che2014}
\begin{equation}\label{CUSBD36}
\min_{\left\{\mu_{k},q_{k},\theta_{k},\kappa_{k},\vartheta_{k}\right\}} \tau\sum\limits_{k\in\mathcal{K}}\mu_{k}+\tau\left(\sum\limits_{k\in\mathcal{K}}\mu_{k}\right)^{2}-
\psi\left(\bm{\kappa},\bm{\mu}\right),~\mathrm{s.t.}~\eqref{CUSBD33c}-\eqref{CUSBD33i},
\end{equation}
where $\psi\left(\bm{\kappa},\bm{\mu}\right)=\sum\limits_{k\in\mathcal{K}}\left(\kappa_{k}^{2}+\tau\mu_{k}^{2}\right)+\tau\left(\sum\limits_{k\in\mathcal{K}}\mu_{k}\right)^{2}$. Note that problem~\eqref{CUSBD36} can be regarded as partial Lagrangian function of problem~\eqref{CUSBD33} with $\tau$ being the Lagrange multiplier associated with constraint~\eqref{CUSBD33b}.

Note that the objective function, constraints~\eqref{CUSBD33c} and~\eqref{CUSBD33d} are non-convex. Therefore, we need to transform them into convex forms via proper operations. Here, we adopt the first-order Taylor series expansion approximation of functions $\psi\left(\bm{\kappa},\bm{\mu}\right)$, $\phi_{k}\left(\theta_{k},\mathbf{q}\right)$, and $\tilde{\phi}_{k}\left(\mu_{k},\mathbf{q}\right)$, around points $\left(\bm{\kappa}^{\left(t\right)},\bm{\mu}^{\left(t\right)}\right)$, $\left(\theta_{k}^{\left(t\right)}, \mathbf{q}^{\left(t\right)}\right)$, and $\left(\mu_{k}^{\left(t\right)}, \mathbf{q}^{\left(t\right)}\right)$ to obtain a low boundary convex approximation, respectively. Let $\mu_{k}^{\left(t\right)}$, $q_{k}^{\left(t\right)}$, and $\theta_{k}^{\left(t\right)}$ be the values of $\mu_{k}$, $\theta_{k}$ and $q_{k}$ at the $t$-th iteration, respectively, $\forall k\in\mathcal{K}$, and $t$ denotes the index of iteration. Therefore, we have
\begin{subequations}\label{URLLC37}
\begin{align}
\psi\left(\bm{\kappa},\bm{\mu}\right)\geq&\rho\left(\bm{\kappa},\bm{\mu}\right)\triangleq\psi\left(\bm{\kappa}^{\left(t\right)},\bm{\mu}^{\left(t\right)}\right)\\\nonumber
&+2\sum\limits_{k\in\mathcal{K}}\left(\kappa_{k}^{\left(t\right)}\left(\kappa_{k}-\kappa_{k}^{\left(t\right)}\right)+\tau\left(\mu_{k}^{\left(t\right)}+
\sum\limits_{l\in\mathcal{K}}\mu_{l}^{\left(t\right)}\right)\left(\mu_{k}-\mu_{k}^{\left(t\right)}\right)\right),\\
\phi_{k}\left(\theta_{k},\mathbf{q}\right)\geq&\varrho_{k}\left(\theta_{k},\mathbf{q}\right)\triangleq\phi_{k}\left(\theta_{k}^{\left(t\right)},\mathbf{q}^{\left(t\right)}\right)\\\nonumber
&+\theta_{k}^{\left(t\right)}\left(\theta_{k}-\theta_{k}^{\left(t\right)}\right)
+\sum\limits_{l\neq k, l\in\mathcal{K}}\sigma_{k,l}\left(\mathbf{q}^{\left(\tau\right)}\right)\left(q_{l}-q_{l}^{\left(\tau\right)}\right),\\
\tilde{\phi}_{k}\left(\mu_{k},\mathbf{q}\right)\geq&\tilde{\varrho}_{k}\left(\mu_{k},\mathbf{q}\right)\triangleq\tilde{\phi}_{k}\left(\mu_{k}^{\left(t\right)},\mathbf{q}^{\left(t\right)}\right)\\\nonumber
&+\tilde{\gamma}_{k}\mu_{k}^{\left(t\right)}\left(\mu_{k}-\mu_{k}^{\left(t\right)}\right)
+\sum\limits_{l\neq k, l\in\mathcal{K}}\sigma_{k,l}\left(\mathbf{q}^{\left(\tau\right)}\right)\left(q_{l}-q_{l}^{\left(\tau\right)}\right),
\end{align}
\end{subequations}
where $\sigma_{k,m}\left(\mathbf{q}\right)\triangleq\left|\overline{\mathbf{h}}_{m}^{H}\mathbf{w}_{k}\right|^{2}\sum\limits_{n\neq k, n\in\mathcal{K}}q_{n}\left|\overline{\mathbf{h}}_{n}^{H}\mathbf{w}_{k}\right|^{2}$. Thus, we can rewrite problem~\eqref{CUSBD36} into the following low boundary convex form
\begin{subequations}\label{CUSBD38}
\begin{align}
&\min_{\left\{\mu_{k},q_{k},\theta_{k},\kappa_{k},\vartheta_{k}\right\}} \sum\limits_{k\in\mathcal{K}}\tau\mu_{k}-\rho\left(\bm{\kappa},\bm{\mu}\right) \label{CUSBD38a}\\
\mathrm{s.t.}~&\psi_{k}\left(\theta_{k},\mathbf{q}\right)-\varrho_{k}\left(\theta_{k},\mathbf{q}\right)\leq 0, \forall k\in\mathcal{K},\label{CUSBD38b}\\
&\tilde{\psi}_{k}\left(\mu_{k},\mathbf{q}\right)-\tilde{\varrho}_{k}\left(\mu_{k},\mathbf{q}\right)\leq 0, \forall k\in\mathcal{K},\label{CUSBD38c}\\
&\eqref{CUSBD33e}-\eqref{CUSBD33i}.\label{CUSBD37d}
\end{align}
\end{subequations}
Note that $\kappa_{k}^{2}\leq \mu_{k}\vartheta_{k}$ can be rewritten as $\left\|\left[\kappa_{k},\frac{\mu_{k}-\vartheta_{k}}{2}\right]^{T}\right\|\leq\frac{\mu_{k}+\vartheta_{k}}{2}$. This implies that problem~\eqref{CUSBD38} can be easily solved with the classical convex optimization tools, such as the CVX tools~\cite{CVXTool}. The algorithm used to solve problem~\eqref{CUSBD05} is summarized in Algorithm~\ref{CUSBDA01}, where $\varsigma>0$ is a update step-size, $\zeta^{\left(t\right)}$ is the objective value of problem~\eqref{CUSBD38} at the $t$-th iteration, $\upsilon^{\left(\iota\right)}$ is the objective value of problem~\eqref{CUSBD11} at the $\iota$-th iteration with fixed auxiliary variables $\bm{\lambda}$ and $\delta>0$ is a stopping threshold. The convergence of Steps~\ref{URLLCA0103} to~\ref{URLLCA0105} can be guaranteed with the successive convex approximation~\cite{SPLTran2012} with the conclusions obtained~\cite{TCOMKwan2014}. Meanwhile, the convergence of Steps~\ref{URLLCA0103} to~\ref{URLLCA0106} can be guaranteed by using the monotonic boundary theory~\cite{MathBibby1924}. In a word, the convergence of Algorithm~\ref{CUSBDA01} can be guaranteed with the conclusions obtained in~\cite{TITZhang2012}.
\begin{algorithm}[!ht]
\caption{Solution of constrained problem~\eqref{CUSBD05}}\label{CUSBDA01}
\begin{algorithmic}[1]
\STATE Initialize $\lambda_{b}, \forall b\in\mathcal{B}$, $\tau^{\left(0\right)}>0$, $\varrho^{\left(0\right)}=1$, and  stopping threshold $\delta$.\label{URLLCA0100}
\STATE Let $\iota=0$ and $t=0$ and $\tau=\tau^{\left(0\right)}$.  Initialize beamforming vector $\mathbf{w}_{k}^{\left(0\right)}$ and $p_{k}^{\left(0\right)}$, $\forall k\in\mathcal{K}$, such that constraint~\eqref{CUSBD05b} and~\eqref{CUSBD05c} are satisfied. \label{URLLCA0101}
\STATE Compute $\overrightarrow{\gamma}_{k}$ with $\mathbf{w}_{k}^{\left(0\right)}$ and $p_{k}^{\left(0\right)}$ to obtain $\bar{\gamma}_{k}$, $\forall k\in\mathcal{K}$. Compute $q_{k}^{\left(0\right)}$ with~\eqref{CUSBD16}, let $\theta_{k}=\bar{\gamma}_{k}$, $\vartheta_{k}=\log_{2}\left(1+\theta_{k}\right)$, $\kappa_{k}^{2}=\mu_{k}\vartheta_{k}$, $\forall k\in\mathcal{K}$, and initialize $\zeta^{\left(0\right)}$ and $\upsilon^{\left(0\right)}$.\label{URLLCA0102}
\STATE Let $t\leftarrow t+1$. Solve problem~\eqref{CUSBD38} to obtain $\mu_{k}^{\left(t\right)}$, $q_{k}^{\left(t\right)}$, $\theta_{k}^{\left(t\right)}$, $\kappa_{k}^{\left(t\right)}$, $\vartheta_{k}^{\left(t\right)}$, and $\zeta^{\left(t\right)}$, $\forall k\in\mathcal{K}$.\label{URLLCA0103}
\STATE If $\left|\frac{\zeta^{\left(t\right)}-\zeta^{\left(t-1\right)}}{\zeta^{\left(t-1\right)}}\right|\leq\delta$, go to Step~\ref{URLLCA0105}. Otherwise, update $\tau$ as
\begin{equation}\label{CUSBD39}
\tau=\left\{
\begin{aligned}
\tau+\varsigma\sum\limits_{k\in\mathcal{K}}\left(\mu_{k}-\mu_{k}^{2}\right),~\text{if}~\tau+\varsigma\sum\limits_{k\in\mathcal{K}}\left(\mu_{k}-\mu_{k}^{2}\right)>0\\
\tau\qquad\qquad,~\text{if}~\tau+\varsigma\sum\limits_{k\in\mathcal{K}}\left(\mu_{k}-\mu_{k}^{2}\right)<0.
\end{aligned}
\right.
\end{equation}
and go to Step~\ref{URLLCA0103}. \label{URLLCA0104}
\STATE Let $\iota=\iota+1$, update $\mathbf{w}_{k}^{\left(\iota\right)}$ with $q_{k}^{\left(t\right)}$ and~\eqref{CUSBD20}, and calculate the objective value $\upsilon^{\left(\iota\right)}$. If $\left|\frac{\upsilon^{\left(\iota\right)}-\upsilon^{\left(\iota-1\right)}}{\upsilon^{\left(\iota-1\right)}}\right|\leq\delta$, stop iteration and go to Step~\ref{URLLCA0106}. Otherwise, let $\tau=\tau^{\left(0\right)}$ and go to Step~\ref{URLLCA0103}.\label{URLLCA0105}
\STATE Obtain the objective value $\varrho^{\left(*\right)}=\upsilon^{\left(\iota\right)}$, $\left|\frac{\varrho^{\left(*\right)}-\varrho^{\left(0\right)}}{\varrho^{\left(0\right)}}\right|\leq\delta$, then stop iteration. Otherwise, let $\varrho^{\left(0\right)}=\varrho^{\left(*\right)}$ and update $\lambda_{b}$ with~\eqref{CUSBD40}, $\forall b\in\mathcal{B}$,
\begin{equation}\label{CUSBD40}
\lambda_{b}=\left\{
\begin{aligned}
\lambda_{b}\qquad,~\text{if}~\lambda_{b}-\varsigma g_{b}<0\\
\lambda_{b}-\varsigma g_{b},~\text{if}~\lambda_{b}-\varsigma g_{b}\geq 0.
\end{aligned}
\right.
\end{equation}
and go Step~\ref{URLLCA0103}.\label{URLLCA0106}
\end{algorithmic}
\end{algorithm}

\subsection{Brute-force Search Algorithm}

As discussed in the last paragraph of Section~\ref{SystemModelAndProblem}, we can obtain the optimal the solution to problem~\eqref{CUSBD05} or problem~\eqref{CUSBD11} via fixing set $\mathcal{S}$ determined by the brute-force search method. The first step of the brute-force search algorithm is to check the feasibility of the set $\mathcal{S}$, i.e., whether the minimum user rate of each user in the set $\mathcal{S}$ is met or not. The feasibility can be judged via addressing a power minimization problem, i.e.,
\begin{subequations}\label{CUSBD41}
\begin{align}
&\min\limits_{\left\{\tilde{\mathbf{w}}_{k'}\right\}}~\sum\limits_{k'\in\mathcal{S}}\left\|\tilde{\mathbf{w}}_{k'}\right\|^{2}\label{CUSBD41a}\\
\mathrm{s.t.}~&\tilde{\gamma}_{k'}\leq\frac{\left|\overline{\mathbf{h}}_{k'}^{H}\tilde{\mathbf{w}}_{k'}\right|^{2}}
{\sum\limits_{l'\neq k'}\left|\overline{\mathbf{h}}_{k'}^{H}\tilde{\mathbf{w}}_{l'}\right|^{2}+1}, \forall k'\in\mathcal{S}.\label{CUSBD41b}
\end{align}
\end{subequations}
Let $\mathbf{o}=\left[\overline{\mathbf{h}}_{k'}^{H}\tilde{\mathbf{w}}_{1'},\overline{\mathbf{h}}_{k'}^{H}\tilde{\mathbf{w}}_{2'},\cdots,
\overline{\mathbf{h}}_{k'}^{H}\tilde{\mathbf{w}}_{K'},1\right]^{T}$. Thus, problem~\eqref{CUSBD41} can be rewritten as follows:
\begin{equation}\label{URLLD42}
\min\limits_{\left\{\tilde{\mathbf{w}}_{k'}\right\}}~\sum\limits_{k'\in\mathcal{K}}\left\|\tilde{\mathbf{w}}_{k'}\right\|^{2},~
\mathrm{s.t.}~\left\|\mathbf{o}\right\|\leq\sqrt{1+\frac{1}{\tilde{\gamma}_{k'}}}\overline{\mathbf{h}}_{k'}^{H}\tilde{\mathbf{w}}_{k'}.
\end{equation}
Problem~\eqref{URLLD42} can be easily solved by the classical convex optimization tools, such as CVX~\cite{CVXTool}. Let $\left\{\tilde{\mathbf{w}}_{k'}^{\left(*\right)}\right\}$  be the optimal solution to problem~\eqref{URLLD42}. If $\sum\limits_{k'\in\mathcal{S}}\left\|\mathbf{Q}_{b}\tilde{\mathbf{w}}_{k'}\right\|_{2}^{2}\leq P_{b}$, $\forall b\in\mathcal{B}$, and $\left|\mathcal{S}\right|\leq BN_{\mathrm{t}}$, then the minimum user rate of each user in the set $\mathcal{S}$ is met, i.e., the combination of users in the set $\mathcal{S}$ is valid. Further, for a given set $\mathcal{S}$, the $\left\{\mathbf{w}_{k'}^{\left(*\right)}\right\}$ and $\left\{q_{k'}^{\left(*\right)}\right\}$ can be obtained via solving the following problem:
\begin{subequations}\label{CUSBD43}
\begin{align}
&\max_{\left\{q_{k'},\mathbf{w}_{k'}\right\}} \sum\limits_{k'\in\mathcal{S}}\overleftarrow{R}_{k'}, \label{CUSBD43a}\\
\mathrm{s.t.}~&q_{k'}\geq 0, \left\|\mathbf{w}_{k'}\right\|_{2}=1, \forall k'\in\mathcal{S},\label{CUSBD43b}\\
&r_{k'}\leq \overleftarrow{R}_{k'}, \forall k'\in\mathcal{S},\label{CUSBD43c}\\
&\sum\limits_{k'\in\mathcal{S}}q_{k'}\leq \sum\limits_{b\in\mathcal{B}}\lambda_{b}P_{b}.\label{CUSBD43d}
\end{align}
\end{subequations}
As discussed in the last paragraph of Section~\ref{SystemModelAndProblem}, $\sum\limits_{\left|\mathcal{S}\right|=1}^{BN_{\mathrm{t}}}\frac{K!}{\left|\mathcal{S}\right|!\left(K-\left|\mathcal{S}\right|\right)!}$ combinational problems are needed to be addressed and the optimal combination can be selected out to achieve the maximal sum rate. Note that problem~\eqref{CUSBD43} can be solved with a similar method as Algorithm~\ref{CUSBDA01}.In particular, with fixed virtual uplink transmit powers $\left\{q_{k}\right\}$, the optimal beamforming vectors $\left\{\mathbf{w}_{k}\right\}$ is calculated with~\eqref{CUSBD31}. Given beamforming vectors $\left\{\mathbf{w}_{k}\right\}$ and auxiliary variable $\left\{\lambda_{b}\right\}$, $\forall k\in\mathcal{K}$ and $\forall b\in\mathcal{B}$, problem~\eqref{CUSBD43} is reformulated as follows:
\begin{subequations}\label{CUSBD44}
\begin{align}
&\max_{\left\{q_{k'},\theta_{k'}\right\}} \sum\limits_{k'\in\mathcal{S}}\log_{2}\left(1+\theta_{k'}\right), \label{CUSBD44a}\\
\mathrm{s.t.}~&\psi_{k'}\left(\theta_{k'},\mathbf{q}^{'}\right)-\phi_{k'}\left(\theta_{k'},\mathbf{q}^{'}\right)\leq 0, \forall k'\in\mathcal{S},\label{CUSBD44b}\\
&\tilde{\psi}_{k'}^{'}\left(\mathbf{q}^{'}\right)-\tilde{\phi}_{k'}^{'}\left(\mathbf{q}^{'}\right)\leq 0, \forall k'\in\mathcal{S},\label{CUSBD44c}\\
&q_{k'}\geq 0, \forall k'\in\mathcal{S}, \sum\limits_{k'\in\mathcal{S}}q_{k'}\leq \sum\limits_{b\in\mathcal{B}}\lambda_{b}P_{b},\label{CUSBD44d}
\end{align}
\end{subequations}
where $\mathbf{q}^{'}=\left[q_{1'},\cdots,q_{K'}\right]^{T}$, $\tilde{\psi}_{k'}^{'}\left(\mathbf{q}^{'}\right)$ and $\tilde{\phi}_{k'}^{'}\left(\mathbf{q}^{'}\right)$ are defined as
\begin{subequations}\label{CUSBD45}
\begin{align}
\tilde{\psi}_{k'}^{'}\left(\mathbf{q}^{'}\right)&=\widetilde{\gamma}_{k'}\left\|\mathbf{Q}\mathbf{w}_{k'}\right\|_{2}^{2}-q_{k'}\left|\overline{\mathbf{h}}_{k'}^{H}\mathbf{w}_{k'}\right|^{2}
+\tilde{\varphi}_{k'}^{'}\left(\mathbf{q}^{'}\right),\label{CUSBD45a}\\
\tilde{\varphi}_{k'}^{'}\left(\mathbf{q}^{'}\right)&\triangleq\frac{1}{2}\left(\widetilde{\gamma}_{k'}+\sum\limits_{l\neq k',l\in\mathcal{S}}q_{l}\left|\overline{\mathbf{h}}_{l}^{H}\mathbf{w}_{k'}\right|^{2}\right)^{2},\label{CUSBD45b}\\
\tilde{\phi}_{k}^{'}\left(\mathbf{q}^{'}\right)&\triangleq\frac{1}{2}\widetilde{\gamma}_{k}^{2}+\frac{1}{2}\left(\sum\limits_{l\neq k',l\in\mathcal{S}}q_{l}\left|\overline{\mathbf{h}}_{l}^{H}\mathbf{w}_{k'}\right|^{2}\right)^{2}.\label{CUSBD45c}
\end{align}
\end{subequations}
Instead of directly solving problem~\eqref{CUSBD43}, we resort to addressing the following approximated problem:
\begin{subequations}\label{CUSBD46}
\begin{align}
&\max_{\left\{q_{k'},\theta_{k'}\right\}} \sum\limits_{k'\in\mathcal{S}}\log_{2}\left(1+\theta_{k'}\right), \label{CUSBD46a}\\
\mathrm{s.t.}~&\psi_{k'}\left(\theta_{k'},\mathbf{q}^{'}\right)-\rho_{k'}\left(\theta_{k'},\mathbf{q}^{'}\right)\leq 0, \forall k'\in\mathcal{S},\label{CUSBD46b}\\
&\tilde{\psi}_{k'}^{'}\left(\mathbf{q}^{'}\right)-\tilde{\varrho}_{k'}^{'}\left(\mathbf{q}^{'}\right)\leq 0, \forall k'\in\mathcal{S},\label{CUSBD46c}\\
&q_{k'}\geq 0, \forall k'\in\mathcal{S}, \sum\limits_{k'\in\mathcal{S}}q_{k'}\leq \sum\limits_{b\in\mathcal{B}}\lambda_{b}P_{b},\label{CUSBD46d}
\end{align}
\end{subequations}
where $\tilde{\varrho}_{k'}^{'}\left(\mathbf{q}^{'}\right)\triangleq\tilde{\phi}_{k'}^{'}\left(\mathbf{q}^{'\left(t\right)}\right)
+\sum\limits_{l\neq k', l\in\mathcal{S}}\sigma_{k',l}\left(\mathbf{q}^{'\left(\tau\right)}\right)\left(q_{l}-q_{l}^{\left(\tau\right)}\right)$. Consequently, the algorithm used to solve problem~\eqref{CUSBD44} is summarized as Algorithm~\ref{CUSBDA02}.
\begin{algorithm}[!ht]
\caption{Solution of constrained problem~\eqref{CUSBD44}}\label{CUSBDA02}
\begin{algorithmic}[1]
\STATE Initialize $\lambda_{b}, \forall b\in\mathcal{B}$, $\varrho^{\left(0\right)}=1$, and  stopping threshold $\delta$.\label{URLLCA0200}
\STATE Let $t=0$ and $\upsilon^{\left(0\right)}$ to be a nonzero value.  Initialize beamforming vector $\mathbf{w}_{k'}^{\left(0\right)}$ and $q_{k'}^{\left(0\right)}$, $\forall k'\in\mathcal{S}$, such that constraint~\eqref{CUSBD43b},~\eqref{CUSBD43c}, and~\eqref{CUSBD43d} are satisfied. \label{URLLCA0201}
\STATE Let $t\leftarrow t+1$. Solve problem~\eqref{CUSBD46} to obtain $q_{k'}^{\left(t\right)}$, $\theta_{k'}^{\left(t\right)}$, and $\zeta^{\left(t\right)}$, $\forall k'\in\mathcal{S}$.\label{URLLCA0202}
\STATE If $\left|\frac{\upsilon^{\left(t\right)}-\upsilon^{\left(t-1\right)}}{\upsilon^{\left(t-1\right)}}\right|\leq\delta$, stop iteration and go to Step~\ref{URLLCA0204}. Otherwise,  update $\mathbf{w}_{k'}^{\left(t\right)}$ with $q_{k'}^{\left(t\right)}$ and~\eqref{CUSBD20}, and go to Step~\ref{URLLCA0202}.\label{URLLCA0203}
\STATE Obtain the objective value $\varrho^{\left(*\right)}=\upsilon^{\left(t\right)}$, $\left|\frac{\varrho^{\left(*\right)}-\varrho^{\left(0\right)}}{\varrho^{\left(0\right)}}\right|\leq\delta$, then stop iteration. Otherwise, let $\varrho^{\left(0\right)}=\varrho^{\left(*\right)}$ and update $\lambda_{b}$ with~\eqref{CUSBD40} and go Step~\ref{URLLCA0201}.\label{URLLCA0204}
\end{algorithmic}
\end{algorithm}

\subsection{ZFBF with SUS Algorithm}

The process of  searching all possible combinations of scheduling users in the brute-force search algorithm is with high complexity. Hence, in this subsection, we use the semiorthogonal user scheduling (SUS) method to generate the scheduling user set $\mathcal{S}$ with zero-forcing beamforming (ZFBF)~\cite{JSACYoo2006}. For a given set $\mathcal{S}$, beamforming vectors $\mathbf{W}\left(\mathcal{S}\right)=\left[\mathbf{w}_{1'},\cdots,\mathbf{w}_{k'},\cdots,\mathbf{w}_{K'}\right]$ is given by
\begin{equation}\label{CUSBD47}
\mathbf{w}_{k'}=\frac{\mathbf{W}\left(\mathcal{S}\right)\left[:,k'\right]}{\left\|\mathbf{W}\left(\mathcal{S}\right)\left[:,k'\right]\right\|_{2}}
\end{equation}
where $\mathbf{W}\left(\mathcal{S}\right)=\mathbf{H}\left(\mathcal{S}\right)\left(\mathbf{H}\left(\mathcal{S}\right)^{H}\mathbf{H}\left(\mathcal{S}\right)\right)^{-1}$ with $\mathbf{H}\left(\mathcal{S}\right)=\left[\overline{\mathbf{h}}_{1'},\cdots,\overline{\mathbf{h}}_{k'},\cdots,\overline{\mathbf{h}}_{K'}\right]$. Thus, the SINR of the $k'$ user is calculated as
$\gamma_{k'}=\frac{p_{k'}}{\left\|\mathbf{W}\left(\mathcal{S}\right)\left[:,k'\right]\right\|_{2}^{2}}$. To satisfy the minimum rate requirement $r_{k'}$ of the $k'$ user, the required minimum power
is $p_{k'}^{\flat}=\widetilde{\gamma}_{k'}\left\|\mathbf{W}\left(\mathcal{S}\right)\left[:,k'\right]\right\|_{2}^{2}$. If $\sum\limits_{k'\in\mathcal{S}}p_{k'}^{\flat}\left\|\mathbf{Q}_{b}\mathbf{w}_{k'}\right\|_{2}^{2}\leq P_{b}$, $\forall b\in\mathcal{B}$, and $\left|\mathcal{S}\right|\leq BN_{\mathrm{t}}$, the set $\mathcal{S}$ of scheduling users is effectiveness. Furthermore, for a fixed set $\mathcal{S}$ of scheduling users obtained using the ZFBF and the SUS Algorithm, problem~\eqref{CUSBD03} is rewritten as follows:
\begin{subequations}\label{CUSBD48}
\begin{align}
&\max_{\left\{p_{k'}\right\}} \sum\limits_{k'\in\mathcal{S}}\log_{2}\left(1+\frac{p_{k'}}{\left\|\mathbf{W}\left(\mathcal{S}\right)\left[:,k'\right]\right\|_{2}^{2}}\right), \label{CUSBD48a}\\
\mathrm{s.t.}~&p_{k'}>0, \widetilde{\gamma}_{k'}\leq \frac{p_{k'}}{\left\|\mathbf{W}\left(\mathcal{S}\right)\left[:,k'\right]\right\|_{2}^{2}}, \forall k'\in\mathcal{S},\label{CUSBD48b}\\
&\sum\limits_{k'\in\mathcal{S}}p_{k'}\left\|\mathbf{Q}_{b}\mathbf{w}_{k'}\right\|_{2}^{2}\leq P_{b}, \forall b\in\mathcal{B},\label{CUSBD48c}
\end{align}
\end{subequations}
Note that problem~\eqref{CUSBD48} can be easily solved using the CVX tools~\cite{CVXTool}. The detailed algorithm for solving problem~\eqref{CUSBD03} using the ZFBF with SUS Algorithm is described in Algorithm~\ref{CUSBDA03}. Note that in Algorithm~\ref{CUSBDA03}, the process is initialized with $t=1$, $\mathbf{g}_{k}=\overline{\mathbf{h}}_{k}$. The choice of $\zeta$ is discussed in detail in~\cite{JSACYoo2006}.
\begin{algorithm}[!ht]
\caption{Solving problem~\eqref{CUSBD03} using the ZFBF with the SUS Algorithm}\label{CUSBDA03}
\begin{algorithmic}[1]
\STATE Initialize $t=1$, $\mathcal{T}_{t}=\mathcal{K}$, $\mathcal{S}=\mathcal{\phi}$(empty set).\label{URLLCA0300}
\STATE For each user $k\in\mathcal{T}_{t}$, calculate $\mathbf{g}_{k}$ as $\mathbf{g}_{k}=\overline{\mathbf{h}}_{k}-\sum\limits_{l=1}^{t-1}\frac{\overline{\mathbf{h}_{k}^{H}}\mathbf{g}_{\left(l\right)}}{\left\|\mathbf{g}_{\left(l\right)}\right\|}\mathbf{g}_{\left(l\right)}$. \label{URLLCA0301}
\STATE Select the $t$-th user as $t'=\max\limits_{k\in\mathcal{T}_{t}}\left\|\mathbf{g}_{k}\right\|$, $\mathcal{S}=\mathcal{S}\cup\left\{t'\right\}$, $\mathbf{g}_{\left(t\right)}=\mathbf{g}_{\left(t'\right)}$.\label{URLLCA0302}
\STATE If $\sum\limits_{k'\in\mathcal{S}}p_{k'}^{\flat}\left\|\mathbf{Q}_{b}\mathbf{w}_{k'}\right\|_{2}^{2}\leq P_{b}$, $\forall b\in\mathcal{B}$, and $\left|\mathcal{S}\right|\leq BN_{\mathrm{t}}$, where $\mathbf{w}_{k'}$ is obtained via~\eqref{CUSBD47}, then calculate $\mathcal{T}_{t}$ as:\label{URLLCA0303}
\begin{equation}\label{CUSBD49}
\mathcal{T}_{t+1}=\left\{k\in\mathcal{T}_{t},k\neq t'~|~ \frac{\left|\overline{\mathbf{h}}_{l}^{H}\mathbf{g}_{\left(l\right)}\right|}{\left\|\overline{\mathbf{h}}_{\left(k\right)}\right\|\left\|\mathbf{g}_{\left(l\right)}\right\|}<\zeta\right\},
\end{equation}
let $t=t+1$ and go to Step~\ref{URLLCA0302}. Otherwise, go to step~\ref{URLLCA0304}.
\STATE Solve problem~\eqref{CUSBD48} with the set $\mathcal{S}$ of scheduling users.\label{URLLCA0304}
\end{algorithmic}
\end{algorithm}
\begin{remark}\label{Remark04}
Algorithm~\ref{CUSBDA03} can be regarded as an extension of the ZFBF with SUS algorithm developed in~\cite{JSACYoo2006} subjecting to the constrains of minimum rate requirements and maximum allowable transmit power for the downlink of multicell multiuser joint transmission networks. Note that Algorithm~\ref{CUSBDA01} addresses the joint user scheduling and beamforming design problem, i.e, simultaneously realize the set $\mathcal{S}$ of scheduling users selection and the beamforming vectors optimization. However, Algorithm~\ref{CUSBDA02} and Algorithm~\ref{CUSBDA03} all consider the user scheduling and beamforming design problems, but they are divided into two stages, i.e., selecting the set $\mathcal{S}$ of scheduling users and optimizing the beamforming vectors. Compared to Algorithm~\ref{CUSBDA01} and Algorithm~\ref{CUSBDA02}, the implementation of Algorithm~\ref{CUSBDA03} is relatively simple.
\end{remark}

\section{\label{Simulation} Numerical Results}
In this section, we focus on evaluating the performance of the proposed algorithm via numerical methods. We consider a coordinated cluster of $B = 3$ hexagonal adjacent cells each consisting of one BS and $K_{b}$ users, the total number of users is $K=\sum\limits_{b\in\mathcal{B}}K_{b}$. The $K$ users are randomly distributed in the cooperative region of coordinated cluster, as illustrated in Fig.~\ref{Simulationmodel}. In our experiment, the cell radius is 300m and cooperative radius is 100m. The minimum rate $r_{k}$ of the $k$-th user is set to be $0.3 \widetilde{r}_{k}$ where $\widetilde{r}_{k}$ is the rate of the $k$-th user achieved by single user communication with full power maximum ratio transmission. Stopping threshold $\delta = 10^{-3}$. The simulated channel coefficients follow the previous discussed system model in Section \ref{SystemModelAndProblem}.
All users share the same noise variance, \emph{i.e.}, $\sigma_{k}^{2}=\sigma^{2}$, $\forall k\in\mathcal{K}$.  For easy of notation, we define the SNR as $\mathrm{SNR}=10\log_{10}\left(\frac{P}{\sigma^{2}}\right)$ in decibels. All the experiments are implemented with Matlab R2017a and the configuration of computer is Windows 10, Intel(R) Core(TM) i7-8700, 8 GB RAM.

We first give the computational complexities of Algorithm 1, 2, and 3, which are respectively denoted as
${\Im _1} = O(I{t_1}({(5K)^3}(7K + 2) + K{(B{N_t})^{2.7}}))$,
${\Im _2} = \sum\limits_{\left| S \right| = 1}^{\min (B{N_{\rm{t}}},K)} {\frac{{K!}}{{\left| S \right|!\left( {K - \left| S \right|} \right)!}}} O({(\left| S \right|B{N_t})^3} + {\xi _1}I{t_2}({(2\left| S \right|)^3}(3\left| S \right| + 1) + \left| S \right|{(B{N_t})^{2.7}}))$,
and ${\Im _3} = \sum\limits_{K' = 1}^{\min (B{N_{\rm{t}}},K)} {O({{(K - K')}^2}B{N_t} + {\xi _2}I{t_3}{{(K')}^3}(2K' + B))} $. Here, $O(\cdot)$ stands
for the big-O notation. $It_1, It_2, $ and $It_3$ represent the numbers of the operation times in Algorithm 1, 2, and 3. $\xi_1, \xi_2 = 1$ if the selected users are respectively feasible in Algorithm 2 and 3, and they are 0 vice versa. Algorithms 1 and 3 possess computational complexity with polynomial forms, and computational complexity of Algorithm 3 is lower than that of Algorithm 1. On the other hand, Algorithm 2 achieves the largest computational complexity with a non-polynomial form.

\begin{figure*}[t]
\centering
\includegraphics[width=0.8\columnwidth,keepaspectratio]{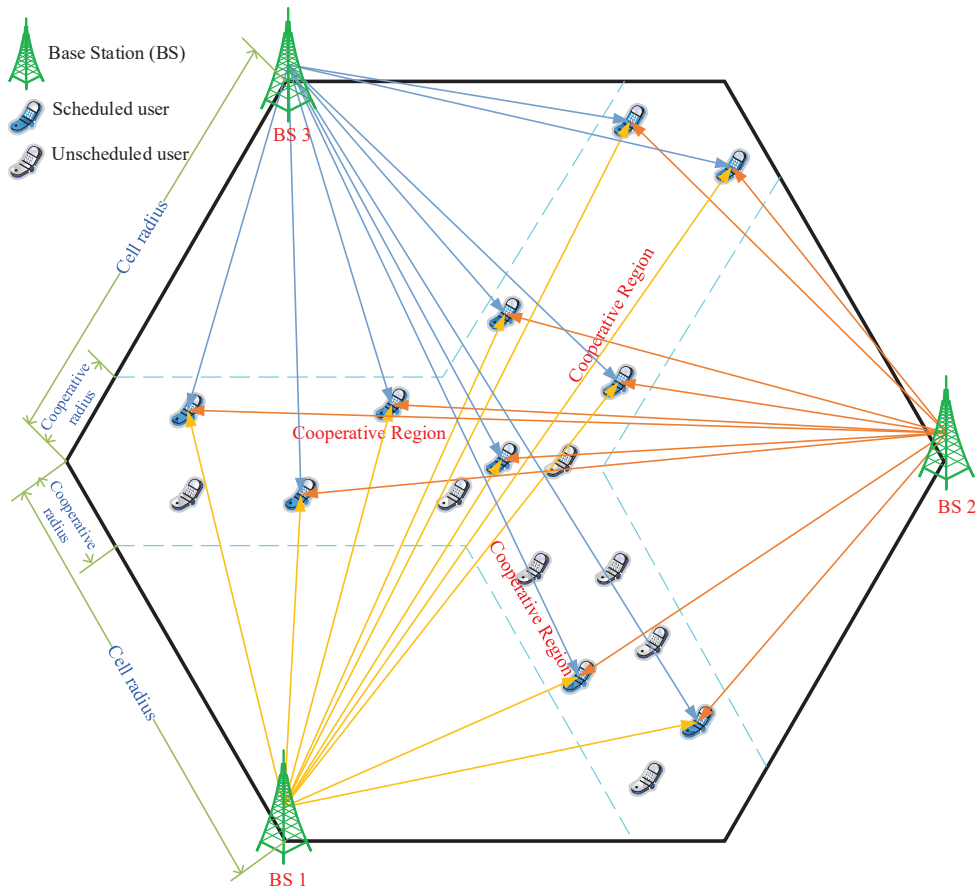}
\caption{Illustration of simulation model.}
\label{Simulationmodel}
\end{figure*}

\subsection{\label{Effectiveness} Effectiveness experiments}
This experiment aims at validating the effectiveness of the proposed Algorithms 1, 2 and 3. To begin with, we generate channel coefficients of user groups with different transmitting parameters. In the experiment, $N_t = 2,  B$ = 3, SNR = 0 dB, $K = 4, 6, 8$. For each group, maximum ratio transmission (MRT) method is first adopted to obtain the minimum rate $r_{k}$, then the three algorithms are individually implemented with the simulated data. Note that monte carlo approach is applied in the experiment and the final sum rate is the mean value.

\begin{figure*}[t]
\centering
\includegraphics[width=1\columnwidth,keepaspectratio]{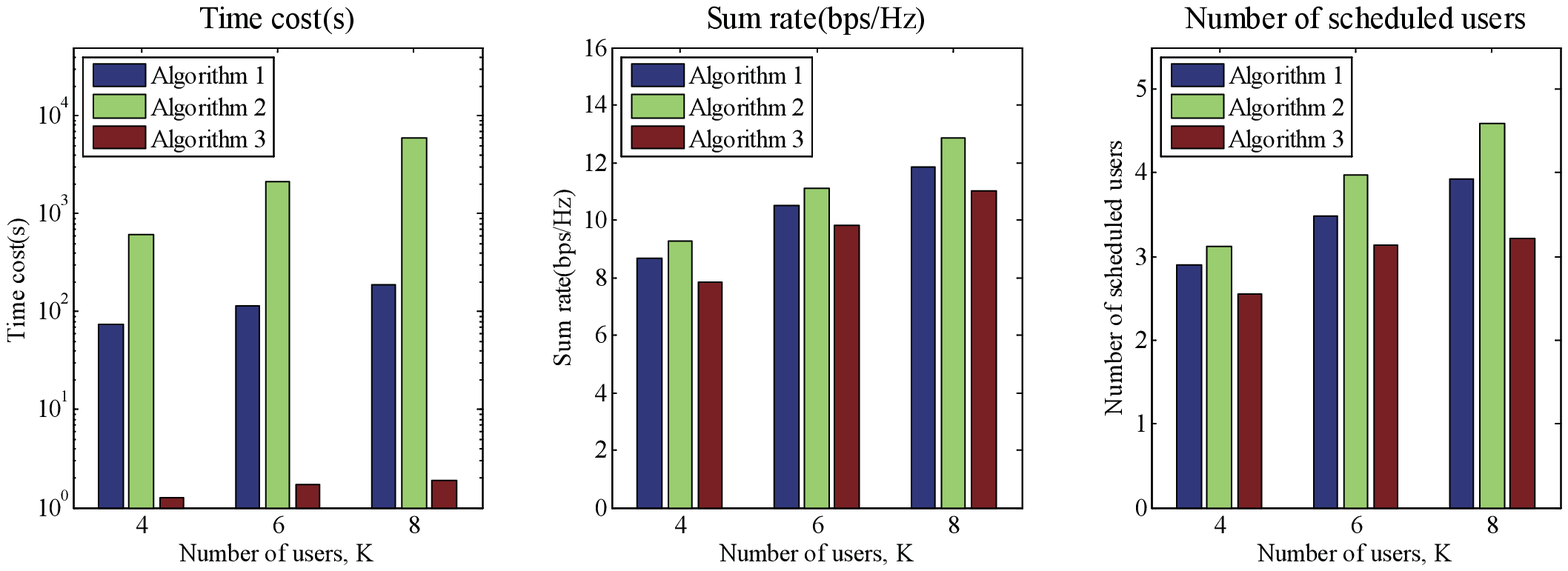}
\caption{Experiment results of different algorithms.}
\label{exp1_res}
\end{figure*}


As it can be seen from Fig. \ref{exp1_res}, all the three algorithms could accomplish the cross-layer optimization tasks, which shows their effectiveness, and Algorithm 1 obtains more balanced results compared with Algorithm 2 and Algorithm 3 in terms of time cost, sum rate and number of scheduled users. Firstly, from the left figure, one can find that the time cost of Algorithm 2 is significantly larger (logarithmic coordinate is applied in y-axis) than Algorithm 1 and Algorithm 3 when the number of users increases from 4 to 8. Algorithm 3 achieves the least time cost since the calculation complexity of power allocation problem (\ref{CUSBD48}) is much simpler than that of Algorithm 1. Generally, the experiment time costs are consistent with the above discussion of computational complexity. Secondly, in the middle figure of Fig. \ref{exp1_res}, Algorithm 2 obtains the global optimal solution at the cost of huge  computation, while Algorithm 3 is inferior to the other two algorithms, becaure the greedy search is easy to fall into local extremum. Algorithm 2, on the other hand, gets a relatively balanced result. Finally, the experimental results also reveal an interesting fact that the numbers of scheduled users using Algorithm 1 and Algorithm 3 are respectively larger and smaller than the optimal value of Algorithm 2.

\subsection{\label{user} Sum rate versus number of users}
\begin{figure*}[t]
\centering
\includegraphics[width=0.7\columnwidth,keepaspectratio]{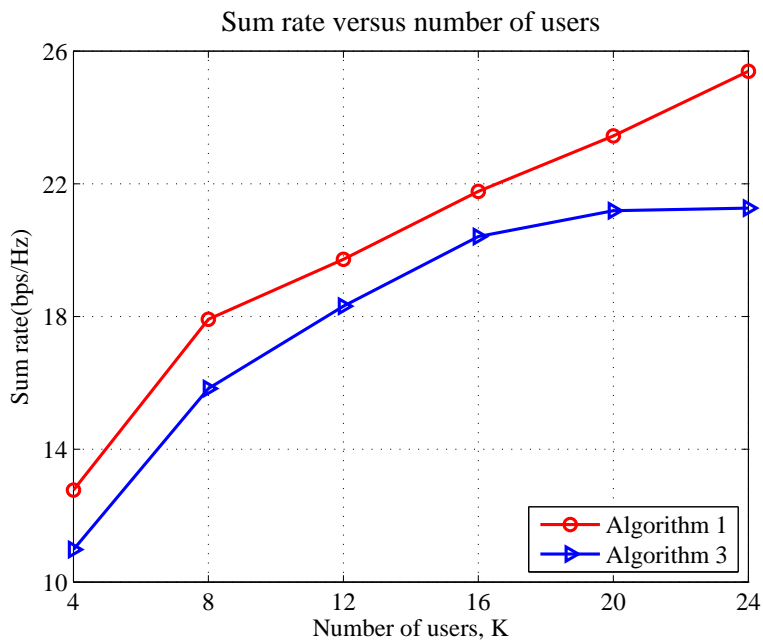}
\caption{Sum rate versus number of users.}
\label{exp2_res}
\end{figure*}
This experiment discusses the performance of Algorithm 1 and 3 with different number of users, where $B = 3, N_t = 4,$ SNR = 0 dB,  and $K =$ 4, 8, 12, 16, 20, 24. Here, Algorithm 2 is not compared since the unbearable time costs. All the results are illustrated in Fig. \ref{exp2_res}. In Fig. \ref{exp2_res}, one can easily find that both the sum rates of Algorithm 1 and Algorithm 3 increase gradually with the increase of user number. However, there is a certain difference in the increasing trend of the two algorithms. As the number of users increases, especially when the number of users is greater than $BN_t$, Algorithm 3 seems to more easily achieve local convergence and the performance gain gradually decreases. Algorithm 1, on the other hand, shows better ability to obtain global optimized results and higher sum rates with the increase of users.

\subsection{\label{SNR} Sum rate versus SNRs}

\begin{figure*}[t]
\centering
\includegraphics[width=1\columnwidth,keepaspectratio]{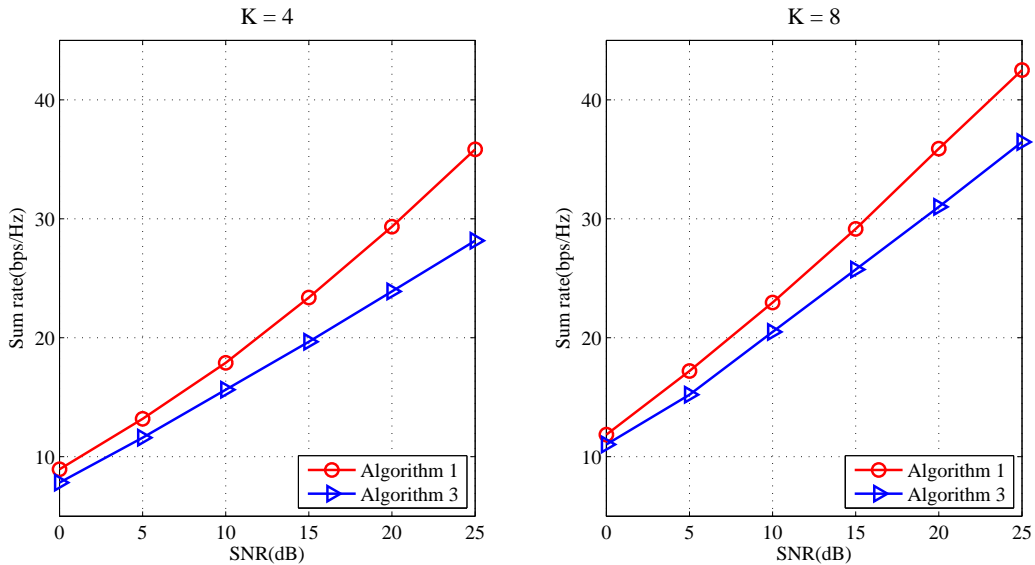}
\caption{Sum rate versus SNRs.}
\label{exp3_res}
\end{figure*}

Fig. \ref{exp3_res} illustrates the average sum rate of Algorithm 1 and Algorithm 3 with different SNRs. In the experiment, $B = 3, N_t = 2, K = 4, 8, 12$,  and SNR = 0, 5, 10, 15, 20, 25 dB. Channel coefficients under the worst SNR condition (SNR = 0 dB) are first generated to confirm the feasibility of users in each group, then the two algorithms are implemented to obtain results. In the figure, numerical results demonstrate that in the low SNR region such as 0$ \sim $10 dB, Algorithm 1 shows slight advantages compared with Algorithm 3, while this advantage increases significantly with the increase of SNR (SNR $\geq$ 15 dB). This is because Algorithm 3 cannot ensure that the global optimal/suboptimal solution is obtained. When the channel condition becomes better (SNR $\geq$ 15 dB), that is, when the solution space is enlarged, local optimal solution is more likely to be obtained using Algorithm 3. However, Algorithm 1 tries to solve the original global optimization problem, so its performance is better than Algorithm 3 when SNR is large.

\subsection{\label{antennas} Sum rate versus number of antennas}

This experiment investigates the performance of Algorithm 1 and 3 with different number of antennas. Fig. \ref{exp4_res} illustrates sum rate versus number of antennas, where $K = 12, 20$, $B =  3, $ SNR = 0, and $N_t = 4, 8, 16$. By observing the experimental results, we can draw two conclusions. On the one hand, all the results become better with the increase of $N_t$, when $K$ is 12 and 20. This indicates that the increase of antennas leads better performance for the joint user scheduling and beamforming design with multiple base stations, since more antennas bring spatial multiple gain. On the other hand, the results obtained by Algorithm 1 are better than those obtained using Algorithm 3, which also implies the advantages of the proposed optimization algorithm.

\begin{figure*}[t]
\centering
\includegraphics[width=1\columnwidth,keepaspectratio]{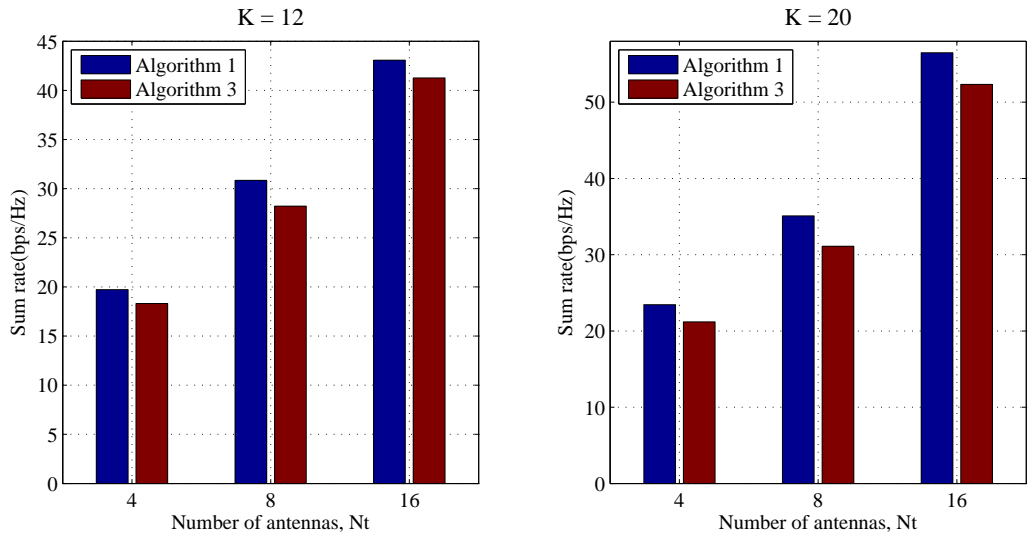}
\caption{Sum rate versus number of antennas.}
\label{exp4_res}
\end{figure*}




\section{\label{Conclusion} Conclusions}
In this paper, we focus on investigating the joint user scheduling and beamforming design for the downlink of multicell multiuser joint transmission networks, while subjecting to the requirement of QoS and the maximum allowable transmit power. For this cross-layer optimization topic,  a mixed discrete-continue variables combinational optimization problem was first formulated. To overcome the difficulties encountered in solving the problem of interest, some basic transformation methods were derived to reformulated the original problem into a tractable form. Then, an effective and efficient optimization method was developed to address it. We also demonstrated that the proposed optimization approach could avoid to allocating power for the unscheduled users. The other two methods, namely, the brute-force search based approach and greedy search based approach, are also proposed. Finally, a large number of experimental results were provided to validate the effectiveness of the developed algorithm, and the proposed optimization approach obtained balanced result compared with the other two approaches.


\begin{thebibliography}{30}
\bibitem{OJWCJiang2021}W. Jiang, B. Han, M. Habibi and H. Schotten, ``The road towards 6G: A comprehensive survey," \emph{IEEE Open J. of the Commun. Society}, vol. 2, pp. 334-366, 2021.
\bibitem{ProHe2021}S. He, Y. Zhang, J. Wang, J. Zhang, J. Ren, Y. Zhang, W. Zhuang, and X. (Sherman) Shen, ``A survey of millimeter-wave communication: Physical-layer technology specifications and enabling transmission technologies," \emph{Proc. of the IEEE}, vol. 109, no. 10, pp. 1666-1705, Oct. 2021.


\bibitem{MagIrmer2011}R. Irmer, \emph{et al.}, ``Coordinated multipoint: Concepts, performance, and field trial results," \emph{IEEE Commun. Mag.,} no. 2, pp. 102-111, 2011.
\bibitem{ProTataria2021}H. Tataria, M. Shafi, A. Molisch, M. Dohler, H. Sj\"{o}land and F. Tufvesson, ``6G wireless systems: Vision, requirements, challenges, insights, and opportunities," \emph{Proc. of the IEEE}, vol. 109, no. 7, pp. 1166-1199, Jul. 2021.

\bibitem{TSPDimic2005}G. Dimic and N. D. Sidiropoulos, ``On downlink beamforming with greedy user selection: performance analysis and simple new algorithm," \emph{IEEE Trans. Signal Process.}, vol. 53, no. 10, pp. 3857-3868,Oct. 2005.

\bibitem{TWCZhang2009}J. Zhang, R. Chen, J. G. Andrews, A. Ghosh, and R. W. Heath, ``Networked MIMO with clustered linear precoding," \emph{IEEE Trans. Wireless Commun.}, vol. 8, no. 4, pp. 1910-1921, Apr. 2009.

\bibitem{JSACYoo2006}T. Yoo and A. Goldsmith, ``On the optimality of multianntenna broadcast scheduling using zero-forcing beamforming," \emph{IEEE J. Sel. Areas Commun.,} vol. 24, no. 3, pp. 528-541, Mar. 2006.

\bibitem{TSPYu2007}W. Yu and T. Lan, ``Transmitter optimization for the multi-antenna downlink with per-antenna power constraints," \emph{IEEE Trans. Signal Process.}, vol. 55, no. 6, pp. 2646-2660, Jun. 2007.

\bibitem{TITHuh2012}H. Huh, A. Tulino, and G. Caire, ``Network MIMO with linear zero-forcing beamforming: Large system analysis, impact of channel estimation, and reduced-complexity scheduling," \emph{IEEE Trans. Infor. Theory}, vol. 58, no. 5, pp. 2911-2934, May 2012.
\bibitem{TWCDah2010}H. Dahrouj and W. Yu, ``Coordinated beamforming for the multicell multi-antenna wireless system," \emph{IEEE Trans. Wireless Commun.}, vol. 9, no. 5, pp. 1748-1759, May 2010.
\bibitem{TITZhang2012}L. Zhang, R. Zhang, Y. Liang, Y. Xin, and H. Poor, ``On Gaussian MIMO BC-MAC duality with multiple transmit covariance constraints," \emph{IEEE Trans. Inf. Theory}, vol. 58, no. 4, pp. 2064-2078, Apr. 2012.
\bibitem{TCOMHe2015}S. He, Y. Huang, L. Yang, B. Ottersten, and W. Hong, ``Energy efficient coordinated beamforming for multicell system: Duality-based algorithm design and massive MIMO transition," \emph{IEEE Trans. Commun.}, vol. 63, no. 12, pp. 4920-4935, Dec. 2015.

\bibitem{CSTFu2014}B. Fu, Y. Xiao, H. Deng and H. Zeng, ``A survey of cross-layer designs in wireless networks," \emph{IEEE Commun. Surveys Tuts.}, vol. 16, no. 1, pp. 110-126, First Quarter 2014.
\bibitem{CLSun2021}S. Sun and S. Moon, ``Practical scheduling algorithms with contiguous resource allocation for next-generation wireless systems," \emph{IEEE Wireless Commun. Lett.}, vol. 10, no. 4, pp. 725-729, Apr. 2021.

\bibitem{TWCNasir2021}A. Nasir, H. Tuan, H. Nguyen, M. Debbah, and H. Poor, ``Resource allocation and beamforming design in the short blocklength regime for URLLC," \emph{IEEE Trans. Wireless Commun.}, vol. 20, no. 2, pp. 1321-1335, Feb. 2021.


\bibitem{JSACHong2013}M. Hong, R. Sun, H. Baligh, Z. Luo, ``Joint base station clustering and beamformer design for partial coordinated transmission in heterogeneous networks," \emph{IEEE J. Sel. Areas Commun.}, vol. 31, no. 2, pp. 226-240, Feb. 2013.
\bibitem{TWCZhang2017}C. Zhang, Y. Huang, Y. Jing, S. Jin and L. Yang, ``Sum-rate analysis for massive MIMO downlink with joint statistical beamforming and user scheduling," \emph{IEEE Trans. wireless Commun.}, vol. 16, no. 4, pp. 2181-2194, Apr. 2017.
\bibitem{TWCJiang2018}Z. Jiang, S. Chen, S. Zhou and Z. Niu, ``Joint user scheduling and beam selection optimization for beam-based massive MIMO downlinks," \emph{IEEE Trans. Wireless Commun.}, vol. 17, no. 4, pp. 2190-2204, Apr. 2018.


\bibitem{TCOMAntoniolo2020}R. Antonioli, G. Fodor, P. Soldati and T. Maciel,``Decentralized user scheduling for rate-constrained sum-utility maximization in the MIMO IBC," \emph{IEEE Trans. on Commun.}, vol. 68, no. 10, pp. 6215-6229, Oct. 2020.

\bibitem{TWCKhan2020}A. Khan, R. Adve and W. Yu, ``Optimizing downlink resource allocation in multiuser MIMO networks via fractional orogramming and the hungarian algorithm," \emph{IEEE Trans. Wireless Commun.}, vol. 19, no. 8, pp. 5162-5175, Aug. 2020

\bibitem{SysAkhtar2021}J. Akhtar, K. Rajawat, V. Gupta and A. K. Chaturvedi, ``Joint user and antenna selection in massive-MIMO systems with QoS-constraints," \emph{IEEE Syst. J.}, vol. 15, no. 1, pp. 497-508, Mar. 2021.


\bibitem{CLHe2021}S. He, J. Yuan, Z. An, Y. Yi and Y. Huang, ``Maximizing the set cardinality of users scheduled for ultra-dense uRLLC networks," \emph{IEEE Commun. Lett.}, vol. 25, no. 12, pp. 3952-3955, Dec. 2021.


\bibitem{TVTSchubert2004}M. Schubert and H. Boche, ``Solution of the multiuser downlink beamforming problem with individual SINR constraints," \emph{IEEE Trans. Veh. Technol.}, vol. 53, no. 1, pp. 18-28, Jan. 2004.



\bibitem{BookBorwein2006}J. Borwein and A. Lewis, Convex analysis and nonlinear optimization: Theory and examples. Berlin, Germany: Springer-Verlag, 2006.

\bibitem{JCand2008}E. Cand\`{e}s, M. Wakin, and S. Boyd, ``Enhancing sparsity by reweighted $\ell_{1}$ minimization," \emph{J. Fourier Anal. Appl.}, vol. 14, nos. 5-6, pp. 877-905, May 2008.
\bibitem{}J. Lee and E. Ekici, ``User scheduling and beam alignment in mmWave networks with a large number of mobile users," \emph{IEEE Trans. Wireless Commun.}, vol. 20, no. 10, pp. 6481-6492, Oct. 2021.

\bibitem{Che2014}E. Che, H. D. Tuan and H. H. Nguyen, ``Joint optimization of cooperative beamforming and relay assignment in multi-user wireless relay networks," \emph{IEEE Trans. Wireless Commun.}, vol. 13, no. 10, pp. 5481-5495, Oct. 2014.
\bibitem{SPLTran2012} L. Tran, M. Hanif, A. Tolli, and M. Juntti, ``Fast converging algorithm for weighted sum rate maximization in multicell MISO downlink," \emph{IEEE Signal Process. Lett.}, vol. 19, no. 12, pp. 872-875, Dec. 2012.
\bibitem{TCOMKwan2014}D. Kwan Ng, Y. Wu, and Robert Schober, ``Power efficient resource allocation for full-duplex radio distributed antenna networks," \emph{IEEE Trans. Commun.}, vol. 15, no. 4, pp. 2896-2911, Apr. 2014.
\bibitem{MathBibby1924} J. Bibby, ``Axiomatisations of the average and a further generalisation of monotonic sequences," \emph{Glasgow Math. J.}, vol. 15, pp. 63-65,1974.
\bibitem{CVXTool}CVX: Matlab Software for Disciplined Convex Programming, \emph{http://cvxr.com/cvx/}.

\bibitem{3GPP} 3GPP, ``TR25.996 Spatial channel model for multiple input multiple output (MIMO) simulations (release 8)," \emph{3rd Generation Partnership Project, Tech. Rep.} Jan. 2009.


\end{thebibliography}
\begin{small}

\end{small}
\end{document}